\documentclass[twocolumn]{autart}  % Enable this line and disable the 
\usepackage{graphicx}     % Include this line if your 
\usepackage{xcolor}
\usepackage{amsmath}
\let\theoremstyle\relax
\usepackage{amsthm}
\usepackage[english]{babel}
\usepackage{amssymb}
\usepackage[numbers]{natbib}

\usepackage[justification=centering]{caption}
\newtheorem{definition}{Definition}

\newcommand{\bx}{\mathbf{x}}
\newcommand{\bbf}{\mathbf{f}}
\newcommand{\bg}{\mathbf{g}}
\newcommand{\bu}{\mathbf{u}}

\begin{document}

\begin{frontmatter}
%\runtitle{Insert a suggested running title} % Running title for regular 

\title{Designing Barrier Functions for Graceful Safety Control\thanksref{footnoteinfo}} % Title, preferably not more 

\thanks[footnoteinfo]{This paper was not presented at any IFAC meeting. Corresponding author: Hosam K. Fathy. Tel. +1(301)405-6617. \textit{Email addresses}: $^{\rm a}$ymoon@umd.edu, $^{\rm b}$orosz@umich.edu, $^{\rm c}$hfathy@umd.edu}

\author[1]{Yejin Moon},  % Add the 
\author[2]{G\'abor Orosz},        % e-mail address 
\author[3]{Hosam K. Fathy} % (ead) as shown

\address[1]{Department of Mechanical Engineering, University of Maryland, College Park, MD 20742, USA} % Please supply                       
\address[2]{Department of Mechanical Engineering and Department of Civil and Environmental Engineering, University of Michigan, MI 48109, USA.}       % full addresses
\address[3]{Department of Mechanical Engineering, University of Maryland, College Park, MD 20742, USA}    % here.

\begin{keyword}            % Five to ten keywords, 
        % chosen from the IFAC 
Control barrier functions, graceful safety control, nonlinear systems, collision avoidance \end{keyword}               % keyword list or with the 
                     % help of the Automatica 
                     % keyword wizard

\begin{abstract}  
This paper examines the problem of achieving ``grace" when controlling dynamical systems for safety, which is defined in terms of providing multi-layered safety assurances. 
Namely, two safety layers are created: a primary layer that represents a desirable degree of safety, and a secondary failsafe layer. 
Graceful control then involves ensuring that even if the primary layer is breached, the failsafe layer remains forward invariant. 
The paper pursues this goal by constructing a safety constraint that combines the concepts of zeroing and reciprocal control barrier functions with regard to the primary and secondary safe sets, respectively. 
This constraint is analogous to a stiffening spring, making it possible to construct energy-based analytical proofs of the resulting graceful safety guarantees. 
The proposed approach is developed for systems with a relative degree of either 1 or 2, the latter case being particularly useful for mechanical systems. 
We demonstrate the applicability of the method using a wall collision avoidance example. 
This demonstration highlights the benefits of the proposed approach compared to traditional benchmarks from the literature. 
\end{abstract}

\end{frontmatter}

%%%%%%%%%%%%%%%%%%%%%%%%%%%%%%%%%%%%%%%%%%%%%%%%%%%%%%%%%%%%%%
\section{Introduction}
%%%%%%%%%%%%%%%%%%%%%%%%%%%%%%%%%%%%%%%%%%%%%%%%%%%%%%%%%%%%%%

This paper addresses the challenge of controlling dynamical systems to ensure \textit{safety} and \textit{grace} simultaneously. 
The research is motivated by the growing use of control barrier function (CBF) methods to ensure safety. 
These methods work by ensuring the \textit{forward invariance} of user-defined safe sets in state space, meaning that if a given system is initialized within a safe set, it is guaranteed to remain within it. 
Practical examples of safety-critical control include avoiding unsafe inter-vehicle spacing on highways \cite{molnar2023safety}, excessive temperatures in battery packs \cite{moon2025graceful}, hazardous airplane flight altitudes \cite{molnar2025collision}, etc. 
Unfortunately, there are many practical scenarios where a given system's safe set is breached due to factors such as large external disturbances, component failures, etc. 
For instance, safe inter-vehicle spacing can be compromised when another vehicle ``cuts in" too close in front of an ego vehicle. Similarly, safe battery pack temperatures can be compromised when a battery cell within a pack catches fire, perhaps due to an event such as nail penetration. 
In such scenarios, a critical need arises for additional \textit{graceful safety} assurances. 
We define \textit{grace}, in this context, as the ability to provide a \textit{multi-layered} safety guarantee, where if the \textit{primary safety} layer is breached, a \textit{secondary failsafe} layer remains forward-invariant. 
One particularly inspiring example is the graceful landing of US Airways Flight 1549 following an engine failure on January 15, 2009 \cite{otfinoski2019captain}. The pilot succeeded in achieving the failsafe goal of landing gracefully in the Hudson river, even when the primary goal of maintaining a safe altitude was no longer viable. 

There is a rich existing literature on CBF-based safety-critical control. 
Two main classes of CBFs are defined, namely, zeroing and reciprocal CBFs~\cite{ames2019control,ames2014control,ames2016control,xu2018constrained}. 
Zeroing and reciprocal CBFs deem a given system ``safe" if a user-defined barrier function is positive or finite, respectively. 
Safety is then ensured by preventing this function from becoming negative or infinite over time, respectively. 
Building on these core foundations, the literature explores fundamental topics such as ensuring finite-time convergence to safety~\cite{li2018formally} as well as designing exponential~\cite{nguyen2016exponential}, high-order~\cite{xiao2019control}, backstepping~\cite{cohen2024safety}, robust~\cite{gurriet2018towards,lopez2020robust,alan2025generalizing}, adaptive~\cite{taylor2020adaptive,xiao2021adaptive}, and stochastic~\cite{clark2019control,fan2020bayesian,santoyo2021barrier,sarkar2020high} CBFs. 
The literature also explores the intersection between safety-critical control and machine learning~\cite{wang2021learning,xiao2021barriernet,xiao2022differentiable,yaghoubi2020training}, reinforcement learning~\cite{cheng2019end,choi2020reinforcement}, episodic learning~\cite{csomay2021episodic,taylor2020learning}, and Bayesian learning~\cite{fan2020bayesian,khojasteh2020probabilistic}. 
In addition to this theoretical research, there is an extensive literature on different applications of safety-critical control. Major application domains include automotive and robotic systems, where CBFs are used for adaptive and connected cruise control~\cite{ames2014control,ames2016control,chinelato2020safe,he2018safety,molnar2023safety}, lane keeping,~\cite{ames2016control,jiang2024safety,xu2017correctness}, safe car racing~\cite{dallas2025control,zeng2021safety}, intersection control~\cite{ma2021model}, connected/automated vehicle control~\cite{alan2023control,he2020improving}, motion control of mobile robots~\cite{desai2022clf}, obstacle avoidance~\cite{panagou2013multi}, robot collision avoidance~\cite{hou2023robust,machida2021consensus}, robot navigation~\cite{agrawal2017discrete,gao2023learning}, and robot traffic merging~\cite{xiao2023barriernet}. 
Beyond these fields, safety control also has applications in areas as diverse as battery management~\cite{hailemichael5110597adaptive,vyas2022thermal} and biomedical systems~\cite{ames2020safety,moon2023safe,yin2025safe}.

Traditional CBF methods have both strengths and limitations. 
On the one hand, these methods typically furnish switching control laws capable of prioritizing other control design objectives, such as reference tracking, when possible~\cite{wieland2007constructive}, while prioritizing safety when it is critical to do so~\cite{goswami2024collision}. 
On the other hand, the existing literature predominantly embraces a single-layer, black-and-white separation between ``safe" and ``unsafe" states. 
Unfortunately, this mindset does not automatically address the need for multi-layered safety guarantees, despite the fact that this need is ubiquitous in practice.  
In the automotive and battery domains, for instance, there is a distinction between ``unsafe" scenarios where inter-vehicle spacing becomes tight or one battery cell catches fire, versus ``catastrophes" where vehicles collide or a cell fire cascades into a multi-cell thermal runaway. 
``Grace" can be intuitively defined as a system's ability to prevent unsafe scenarios from becoming catastrophic~\cite{edwards2017towards,herlihy1991specifying,shelton2003framework,weiss2020fail}. 
Previous work by the authors shows that CBF methods do not always achieve such grace~\cite{moon2024graceful,moon2025graceful}. 

The main goal of this paper is to extend earlier, preliminary research by the authors~\cite{moon2024graceful,moon2025graceful} into a framework for designing graceful safety controllers. 
This framework utilizes a \textit{single} constraint to enforce a \textit{two-layer} definition of safety. 
This produces safety guarantees that are analogous to simultaneously enforcing a zeroing CBF around a primary safety boundary and a reciprocal CBF around a secondary failsafe boundary. 
We prove these guarantees for systems with a relative degree of either 1 or 2. 
This is especially important for mechanical systems, where safety is often defined in terms of position variables, but actuation often involves manipulating acceleration, leading to a relative degree of 2. 

The remainder of this paper is organized as follows. 
Section~\ref{sec:background} surveys some key mathematical foundations from the existing literature. 
Section~\ref{sec:examples} then presents a wall collision example that motivates the need for graceful safety control, especially for systems of relative degree 2. 
Section~\ref{sec:grace} then presents and analyzes the proposed graceful safety control framework, while Section~\ref{sec:simulation} applies this framework to the wall collision avoidance example. Finally, Section~\ref{sec:conclusion} summarizes the paper's conclusions.

%%%%%%%%%%%%%%%%%%%%%%%%%%%%%%%%%%%%%%%%%%%%%%%%%%%%%%%%%%%%%%
\section{Mathematical Background}\label{sec:background}
%%%%%%%%%%%%%%%%%%%%%%%%%%%%%%%%%%%%%%%%%%%%%%%%%%%%%%%%%%%%%%

We begin by introducing background information on reciprocal, zeroing, high-order, and exponential CBFs based on the existing literature, see e.g., \cite{ames2019control,ames2014control,ames2016control,nguyen2016exponential,xiao2019control}. 
Specifically, consider a dynamic system governed by the state-space model: 
\begin{equation} \label{eq.ControlAffineSys}
   \dot{\bx}(t) = \bbf(\bx,\bu),
\end{equation}
where ${\bbf : \mathbb{R}^n \times \mathbb{R}^m \to \mathbb{R}^n}$ is a locally Lipschitz function. 
The symbols ${\bx \in X \subseteq \mathbb{R}^n}$ and ${\bu \in U \subseteq \mathbb{R}^m}$ represent the state and input vectors, respectively, with $X$ and $U$ denoting the sets of admissible states and inputs. 
The literature often makes two assumptions regarding the above dynamics. 
The first assumption is that for any initial state ${\bx(0) \in X}$, there exists a unique solution to the system over a maximum interval of existence, ${I(\bx(0))=[0,\tau)}$. 
This is often referred to as forward completeness.
The second assumption is that the system dynamics are input affine, i.e., the function ${\bbf(\bx,\bu)}$ can be written as ${\bbf(\bx,\bu) = \bbf_{\rm a}(\bx)+\bg_{\rm a}(\bx)\bu}$. 
This is useful for posing safety control as a convex program. 

In order to ensure the safety of the system, one can enforce the forward invariance of a safe set. 
This means that if $\bx(0)$ starts inside of the safe set, $\bx(t)$ will stay inside of the safe set for all ${t\in I(\bx(0))}$. 
In the CBF literature, the safe set ${S \subseteq \mathbb{R}^n}$ is defined as the superlevel set of a continuously differentiable function ${h(\bx):\mathbb{R}^n \rightarrow\mathbb{R}}$. 
In other words, the safe set $S$ consists of state variables for which ${h(\bx)\geq 0}$. 
More precisely, we write
 \begin{equation} \label{eq.SafeSet}
   \begin{split}
   S = \{ \bx \in\mathbb{R}^n : h(\bx) \geq 0 \}, 
   \\
   \partial S = \{ \bx \in\mathbb{R}^n : h(\bx) = 0 \}, 
   \\
   \mathrm{Int}(S) = \{ \bx \in\mathbb{R}^n : h(\bx) > 0 \}. 
   \end{split}
 \end{equation}
That is, the boundary of the safe set $\partial S$ (i.e., the line distinguishing safety vs. lack thereof) is at ${h(\bx)=0}$. 
If the state $\bx$ makes $h(\bx)$ positive, then it lies strictly inside the interior $\mathrm{Int}(S)$ of the set $S$. 
Both reciprocal and zeroing CBFs guarantee the forward invariance of the safe set by appropriately designing an extended class-$\mathcal{K}$ function ${\alpha(r): \mathbb{R} \rightarrow \mathbb{R}}$, i.e., a monotonically increasing function with ${\alpha(0)=0}$.

%%%%%%%%%%%%%%%%%%%%%%%%%%%%%%%%%%%%%%%%%%%%%%%%%%%%%%%%%%%%%%
\subsection{Zeroing Control Barrier Function} \label{sec. zCBF}
%%%%%%%%%%%%%%%%%%%%%%%%%%%%%%%%%%%%%%%%%%%%%%%%%%%%%%%%%%%%%%

A zeroing CBF ensures safe set invariance by changing the sign of the CBF at the safe set boundary~\cite{ames2019control}. 
The definition of zeroing CBF is introduced below.
\\

\begin{definition} \label{def.zCBF}
Consider the control system (\ref{eq.ControlAffineSys}) and the safe set ($\ref{eq.SafeSet}$) for the continuously differentiable function $h(\bx)$. 
Suppose there exist an extended class-$\mathcal{K}$ function $\alpha$ such that
\begin{equation}\label{eq:zCBF}
    \sup_{\bu\in U} \left[
    %\frac{{\rm d}h}{{\rm d}t}
    \dot{h}(\bx,\bu)
    \right] > -\alpha(h(\bx)),
\end{equation}
${\forall \bx \in S}$, then ${h(\bx):\mathbb{R}^n \rightarrow \mathbb{R}}$ is a zeroing CBF. 
\end{definition}

Based on Definition~\ref{def.zCBF} one may formulate the following theorem.
\\

\theoremstyle{theorem}
\newtheorem{theorem}{Theorem}
\begin{theorem}
Consider a continuously differentiable zeroing CBF $h(\bx)$ for the control system \eqref{eq.ControlAffineSys}. 
Then, any locally Lipschitz continuous controller $\bu(\bx)$ that satisfies 
\begin{equation}\label{eq:zCBFset}
  \dot{h}(\bx,\bu)  \geq -\alpha(h(\bx)),
\end{equation}
guarantees the forward invariance of $S$.  
\end{theorem}

One may notice that a strict inequality is used in \eqref{eq:zCBF}
while equality is used in \eqref{eq:zCBFset}.
The former one is required for ensuring Lipschitz continuity of the controller while the latter is used in quadratic programs when synthesizing safety-critical controllers \cite{alan2023control}.

We also remark that most often the linear function ${\alpha(r):=\gamma r}$ is chosen leading to the condition
\begin{equation}\label{eq:reldeg1CBF_exp}
\dot{h}(\bx,\bu) \geq -\gamma h(\bx).
\end{equation}

%%%%%%%%%%%%%%%%%%%%%%%%%%%%%%%%%%%%%%%%%%%%%%%%%%%%%%%%%%%%%%  
\subsection{Reciprocal Control Barrier Function} \label{sec. rCBF}
%%%%%%%%%%%%%%%%%%%%%%%%%%%%%%%%%%%%%%%%%%%%%%%%%%%%%%%%%%%%%%

Reciprocal CBFs are  typically represented by the symbol $B(\bx)$ in literature and are defined inside the safe set (i.e., on $\mathrm{Int}(S)$). 
These prevent the state from leaving the safe set by becoming infinite at the boundary $\partial S$ of the safe set. 
Hence,
\begin{equation}
\inf_{\bx \in\mathrm{Int}(S)} B(\bx) > 0, \quad
\lim_{\bx \rightarrow \partial S} B(\bx) \rightarrow \infty. 
\end{equation}
The definition of the reciprocal CBF used to enforce the forward invariance of $\mathrm{Int}(S)$, is stated below.
\\

\theoremstyle{definition}
\begin{definition} \label{def.rCBF}
Consider the control system (\ref{eq.ControlAffineSys}) and the safe set ($\ref{eq.SafeSet}$) for the continuously differentiable function $h(\bx)$. 
Suppose there exist extended class-$\mathcal{K}$ functions $\alpha_1, \alpha_2, \alpha_3$ such that
\begin{equation}\label{eq:rCBF}
\begin{split}
    &\frac{1}{\alpha_1(h(\bx))}\leq B(\bx) \leq \frac{1}{\alpha_2(h(\bx))},
    \\
    &\inf_{\bu \in U} \left[
    %\frac{{\rm d}B}{{\rm d}t}
    \dot{B}(\bx,\bu)
    \right]\leq \alpha_3\left(h(\bx)\right),
    \end{split}
\end{equation}
${\forall \bx \in\mathrm{Int}(S)}$. 
Then ${B(\bx): \mathrm{Int}(S) \rightarrow \mathbb{R}}$ is a reciprocal CBF.
\end{definition}

%In the above inequality, if $\alpha_3$ and $\frac{\partial B(\bx)}{\partial \bx}$ are locally Lipschitz continuous, $B(\bx)$ is locally Lipschitz continuous as well~\cite{ames2016control}. 
% Given Definition~\ref{def.rCBF}, the set of control inputs that can guarantee the forward invariance of the safe set can be written as follows, assuming an input affine system.
% % \begin{equation}
% %   K_\mathrm{rCBF} (x) = \left\{u\in U: \dot{B}(\bx,\bu) \leq \alpha_3\left(h(\bx)\right) \right\},
% % \end{equation}

Given Definition~\ref{def.rCBF}, one may formulate the following theorem.
\\

\theoremstyle{theorem}
\begin{theorem}
  Consider a continuously differentiable reciprocal CBF $B(\bx)$ for the control system \eqref{eq.ControlAffineSys} on the safe set (\ref{eq.SafeSet}). Then, any locally Lipschitz continuous controller $u$ that satisfies
\begin{equation}
\dot{B}(\bx,\bu) \leq \alpha_3\left(h(\bx)\right) ,
\end{equation}  
  guarantees the forward invariance of $\mathrm{Int}(S)$.  
\end{theorem}

Note that a `canonical' reciprocal CBF can be defined as
\begin{equation}
B(\bx) = \frac{1}{h(\bx)},
\end{equation}
in which case the first line of \eqref{eq:rCBF} trivially holds.

%%%%%%%%%%%%%%%%%%%%%%%%%%%%%%%%%%%%%%%%%%%%%%%%%%%%%%%%%%%%%%
\subsection{High-order and Exponential Control Barrier Functions}
%%%%%%%%%%%%%%%%%%%%%%%%%%%%%%%%%%%%%%%%%%%%%%%%%%%%%%%%%%%%%%

In many applications it occurs that the derivative in \eqref{eq:zCBF}, i.e.,
\begin{equation}
\dot{h}(\bx,\bu) = \nabla h(\bx) \dot{\bx}
= \nabla h(\bx) \bbf(\bx,\bu),
\end{equation}
does not depend on the control input $\bu$, i.e., we have $\dot{h}(\bx)$, rendering the set \eqref{eq:zCBFset} empty and making safety-critical control design infeasible. 
In control affine systems with right hand side ${\bbf(\bx,\bu) = \bbf_{\rm a}(\bx)+\bg_{\rm a}(\bx)\bu}$, this boils down to ${\nabla h(\bx)\bg_{\rm a}(\bx) \equiv \mathbf{0}}$, which is often referred as a relative degree 2 (or higher) system.

This issue can be resolved by defining the high-order CBF \cite{xiao2019control}:
\begin{equation}\label{eq:h2}
h_2(\bx) := \dot{h}(\bx) + \alpha(h(\bx)).
\end{equation}
Then one can ensure ${h_2(\bx)\geq0}$, which holds on the set  
\begin{equation} \label{eq.SafeSet2}
S_2 = \{ \bx \in\mathbb{R}^n : \dot{h}(\bx) + \alpha(h(\bx)) \geq 0 \}, 
\end{equation}
by requiring
\begin{equation}\label{eq:reldeg2CBF}
\dot{h}_2(\bx,\bu) \geq -\alpha_2(h_2(\bx)),
\end{equation}
where $\alpha_2$ is an extended class-$\mathcal{K}$ function.
Then, it can be proven \cite{xiao2019control} that any locally Lipschitz continuous controller ${\bu(\bx)}$ satisfying \eqref{eq:reldeg2CBF} renders the set ${S \cap S_2}$ forward invariant.
Note that here the extended class-$\mathcal{K}$ function $\alpha$ needs to be differentiable.
Moreover, the above method can be extended to higher relative degrees, by differentiating the CBFs until the inputs appear in the derivative.

A special case of high-order CBF is called exponential CBF \cite{ames2019control,nguyen2016exponential}, which is obtained by considering linear functions, that is, ${\alpha(r):=\gamma_1 r}$ and ${\alpha_2(r):=\gamma_2 r}$.
Then, \eqref{eq:h2} and \eqref{eq:reldeg2CBF} lead to
\begin{equation}\label{eq:h2_exp}
h_2(\bx) := \dot{h}(\bx) + \gamma_1 h(\bx),
\end{equation}
and
\begin{equation}\label{eq:reldeg2CBF_exp}
\dot{h}_2(\bx,\bu) \geq -\gamma_2 h_2(\bx),
\end{equation}
yielding
\begin{equation}\label{eq:reldeg2CBF_exp2}
\ddot{h} + (\gamma_1 + \gamma_2) \dot{h} + \gamma_1 \gamma_2 h \geq 0, 
\end{equation}
where we dropped the arguments for simplicity.
One may observe that \eqref{eq:reldeg2CBF_exp2} is a second-order overdamped dynamical system for $h$ with negative characteristic roots $-\gamma_1$ and $-\gamma_2$.

% Exponential CBFs make it possible to extend traditional safety control to problems with relative degree greater than 1. 
% Consider, for example, the dynamic system in \ref{eq.ControlAffineSys}. 
% Suppose that this system's safe set is the superlevel set of some function $h(\bx)$ . 
% Moreover, suppose that the one seeks to ensure safety by enforcing the constraint: 
% \begin{equation}
%     \dot{h} = \frac{{\rm d}h}{{\rm d}x}f(x,u) \geq -\gamma_1 h(\bx), 
% \end{equation}
% \noindent for some user-defined constant $\gamma_1$. Finally, suppose that the system's input does not appear in the above inequality. Then one can define a new function, $h_2$, as follows: 
% \begin{equation}
%     h_2 = \frac{{\rm d}h}{{\rm d}x}f(x,u) + \gamma_1 h(\bx)
% \end{equation}
% One can now enforce the new safety constraint: 
% \begin{equation}
%     \dot{h}_2 \geq -\gamma_2 h_2, 
% \end{equation}
% \noindent for a new user-defined constant, $\gamma_2$. If the system's safety control problem has a relative degree of 2, then the control input will appear in the above inequality. 
% A controller implementing this inequality is then known as an \textit{exponential} safety controller. 

It is important to emphasize that the corresponding safety-critical controller guarantees the forward invariance of ${S \cap S_2}$, not $S$. 
As shown in Section~\ref{sec:examples}, for mechanical system interpretations, this amounts to guaranteeing the forward invariance of a speed-dependent safe distance constraint.

%Consider the nonlinear dynamical system \eqref{eq.ControlAffineSys} with initial condition $x_0$ and the safe set $C\subset \mathbb{R}^n$ \eqref{eq.SafeSet} of an r-times continuously differentiable function $h$. Define the traverse dynamics with traverse vector $\eta(x)$, 
%\[
%\eta (x) = \left\lceil\begin{matrix}
%    h(\bx)\\
%    L_fh(\bx)\\
%    L_f^2h(\bx)\\
%\vdots\\
%    L_f^{r-1}h(\bx)
%\end{matrix}
%\right\rceil.
%\]
%\[
%\dot{\eta}%(x)=F\eta(x)+Gu,~h(\bx)=C\eta(x),
%\]
%where
%\[F = \left\lceil\begin{matrix}
%    0&1&0&\dotsi&0\\
%    0&0&1&\dotsi&0\\
    %\vdots&\vdots&\vdots&~&\vdots\\
    %0&0&0&\dotsi&1\\
    %0&0&0&\dotsi&0\\
%\end{matrix}
%\right\rceil
%\]
%\[
%G = \left\lceil\begin{matrix}
%    0\\
%    0\\
%    \vdots\\
%    0\\
%    1\\
%\end{matrix}
%\right\rceil
%\]
%\[
%C = \left\lceil\begin{matrix}
%    1&0&\dots&0
%\end{matrix}
%\right\rceil
%\]
%\begin{definition}
%A continuously differentiable function $h$ is an exponential CBF for the dynamical system in \eqref{eq.ControlAffineSys} if there exists a row vector $K_b \in \mathbb{R}^r$ s.t.
%\[
%\sup_{u\in U} [L_f^r h(\bx) + L_gL_f^{r-1} h(\bx) u ]\geq -K \eta (x), 
%\]
%$\forall \bx \in\mathrm{Int}(C)$. This results in $h(\bx)\geq Ce^{(F-GK)t}\eta(x_0)\geq0$ whenever $h(x_0)\geq0$. 
    
%\end{definition}

%%%%%%%%%%%%%%%%%%%%%%%%%%%%%%%%%%%%%%%%%%%%%%%%%%%%%%%%%%%%%%
\section{Motivating examples}\label{sec:examples}
%%%%%%%%%%%%%%%%%%%%%%%%%%%%%%%%%%%%%%%%%%%%%%%%%%%%%%%%%%%%%%

This section motivates the paper's graceful safety control approach by presenting two simulation examples highlighting some of the limitations of the above traditional approaches. 
Both examples focus on a simple wall collision avoidance problem, but the intent is to motivate a mathematical framework broadly applicable to other systems. 
Fig.~\ref{fig.object} presents the basic setting for this problem.
An object (e.g., a point mass) with location $x(t)$ translates in a 1-dimensional space, and must avoid collision with a wall at some location $x_\mathrm{w}$. 
A baseline controller attempts to move the object towards a desired location $x_\mathrm{d}$ on the other side of the wall, but this location is infeasible because of the above collision avoidance need. 
In both simulation examples, the primary definition of safety is that the object should ideally be located at a safe distance from the wall, i.e., ${x(t)\geq x_\mathrm{sf}}$, where ${x_\mathrm{sf}>x_\mathrm{w}}$. 
The main difference between the two examples is the choice of mathematical model, and the impact of this choice on relative degree. 
The first example assumes that the velocity $\dot{x}(t)$ of the object can be actuated directly, leading to a first-order differential equation and a relative degree of 1. 
In contrast, the second example assumes that the acceleration $\ddot{x}(t)$ of the object is the control input, leading to a second order differential equation and a relative degree of 2. 

%%%%%%%%%%%%%%%%%%%%%%%%%%%%%%%%%%%%%%%%%%%%%%%%%%%%
\begin{figure}
\begin{center}
\includegraphics[width=0.45\textwidth]{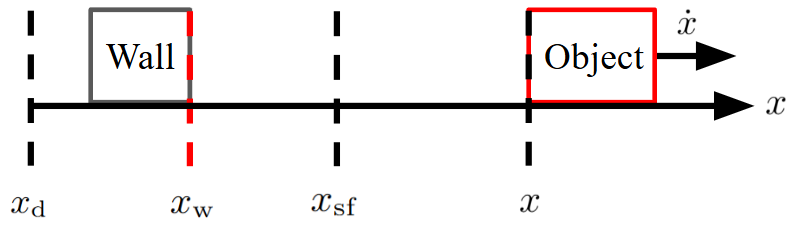}  % The printed column 
\caption{Mechanical model used for the illustrative examples in this paper.} % width is 8.4 cm.
\label{fig.object}                 % Size the figures 
\end{center}                 % accordingly.
\end{figure}
%%%%%%%%%%%%%%%%%%%%%%%%%%%%%%%%%%%%%%%%%%%%%%%%%%%%

%%%%%%%%%%%%%%%%%%%%%%%%%%%%%%%%%%%%%%%%%%%%%%%%%%%%%%%%%%%%%%
\subsection{Example \#1: a first-order zeroing CBF}
%%%%%%%%%%%%%%%%%%%%%%%%%%%%%%%%%%%%%%%%%%%%%%%%%%%%%%%%%%%%%%

Suppose that the velocity of the above object can be manipulated directly, as a control input. 
This corresponds to the first-order differential equation
\begin{equation}\label{eq:1storderODE}
\dot{x} = u.
\end{equation}
For simplicity, omit the input bounds, that is, assume ${u\in\mathbb{R}}$. 
Also, consider the baseline control policy 
\begin{equation} \label{eq.1stOrderBaselineDesiredxdot}
    u_\mathrm{d} = -k(x-x_\mathrm{d}),
\end{equation}
whose goal is to bring the object to the desired target location $x_\mathrm{d}$ using the control gain $k$. 
In the absence of safety-critical control, this will cause a collision between the object and the wall. 
To prevent this, consider the safe set: 
\begin{equation}
 S = \{ x \in\mathbb{R}:x \geq x_\mathrm{sf} \}, 
%    x \geq x_\mathrm{sf},
\end{equation}
and construct the simple CBF candidate:
\begin{equation} 
\label{eq.1stCBF_h}
    h(x) = x-x_\mathrm{sf},
\end{equation}
where the positivity of $h$ implies the object's safety. 
Substituting this into \eqref{eq:reldeg1CBF_exp} yields
\begin{equation} \label{eq.1stOrderBaselineSafexdot}
    u \geq -\gamma(x-x_\mathrm{sf}) =: u_\mathrm{sf},
\end{equation}
where $\gamma$ is a user-defined constant. 

Finally, we construct an optimization problem, where the goal is to minimize the difference between the desired and actual control inputs, subject to the CBF-based safety constraint: 
\begin{equation}\label{eq.1stOrderBenchmark_Optimization}
\begin{split}
&\min_{u\in\mathbb{R}} \frac{1}{2}(u-u_\mathrm{d})^2, 
\\
&{\rm s.t.} \quad u \geq u_\mathrm{sf},
\end{split}
\end{equation}
where $u_\mathrm{d}$ and $u_\mathrm{sf}$ are defined in \eqref{eq.1stOrderBaselineDesiredxdot} and \eqref{eq.1stOrderBaselineSafexdot}. 
Applying the KKT conditions, this quadratic problem (QP) has the analytical solution
\begin{equation}\label{eq.fromKKT1}
\begin{split}
u^* &= \max\{u_\mathrm{d},u_\mathrm{sf}\} 
\\
&= \max\{-k(x-x_\mathrm{d}),-\gamma(x-x_\mathrm{sf})\}.
\end{split}
\end{equation}
Putting this into \eqref{eq:1storderODE} gives the closed-loop dynamics.

%%%%%%%%%%%%%%%%%%%%%%%%%%%%%%%%%%%%%%%%%%%%%%%%%%%%%
\begin{table}[!t]
\centering
    \caption{List of Parameters used in Example \#1} \vspace{4pt} 
 \label{tab:parameters}
    \begin{tabular}{|c|c|c|}
    \hline
    Symbol & Definition & Value 
    \\
    \hline
    $x_\mathrm{d}$ & Desired target location & $0$ ${\rm m}$ 
    \\
    \hline
    $x_\mathrm{w}$ & Wall location & $1$ ${\rm m}$ 
    \\
    \hline
    $x_\mathrm{sf}$ & Safe distance & $3$ ${\rm m}$ 
    \\
    \hline
    $k$ & Baseline control policy gain & $0.5$ ${\rm s}^{-1}$ 
    \\  
    \hline
    $\gamma$ & CBF gain & $3$ ${\rm s}^{-1}$ 
    \\
    \hline
    \end{tabular} 
\end{table}
%%%%%%%%%%%%%%%%%%%%%%%%%%%%%%%%%%%%%%%%%%%%%%%%%%%%%

The above example problem is simulated in MATLAB using the ``ode15s" function, with an 8-second time horizon, a time step of 1 ms, and the parameter values in Table~\ref{tab:parameters}.
Different initial positions of ${x(0)=2}$ m, 5 m, 7 m, and 10 m are used to place the object within the unsafe or safe sets.

Figure~\ref{fig.1stOrderBenchmark1} plots the object's position and velocity for this example while Fig.~\ref{fig.1stOrderBenchmark2} shows the value of the barrier function \eqref{eq.1stCBF_h} and the trajectory in phase space. 
When the object is placed at ${x(0)=2\ {\rm m}}$, due to its unsafe initial state, the controller applies a positive velocity. As a result, the object approaches $x_\mathrm{sf}$, and the barrier function converges exponentially to 0. Second, for ${x(0) = 5\ {\rm m}}$, ${x(0) = 7\ {\rm m}}$ and ${x(0)=10\ {\rm m}}$, the controller first applies the desired velocity from  \eqref{eq.1stOrderBaselineDesiredxdot} to move the object toward the desired position $x_\mathrm{d}$. As the object moves closer to the safety boundary at $x_\mathrm{sf}$, the controller switches to the safe input profile \eqref{eq.1stOrderBaselineSafexdot} to prevent collision. 
The resulting phase space plot is a piecewise linear trajectory with a slope of $-k$ for the desired velocity control and $-\gamma$ for the safe velocity control. 

%%%%%%%%%%%%%%%%%%%%%%%%%%%%%%%%%%%%%%%%%%%%%%%%%%%%% 
\begin{figure} [t]
\begin{center}
\includegraphics[width=0.45\textwidth]{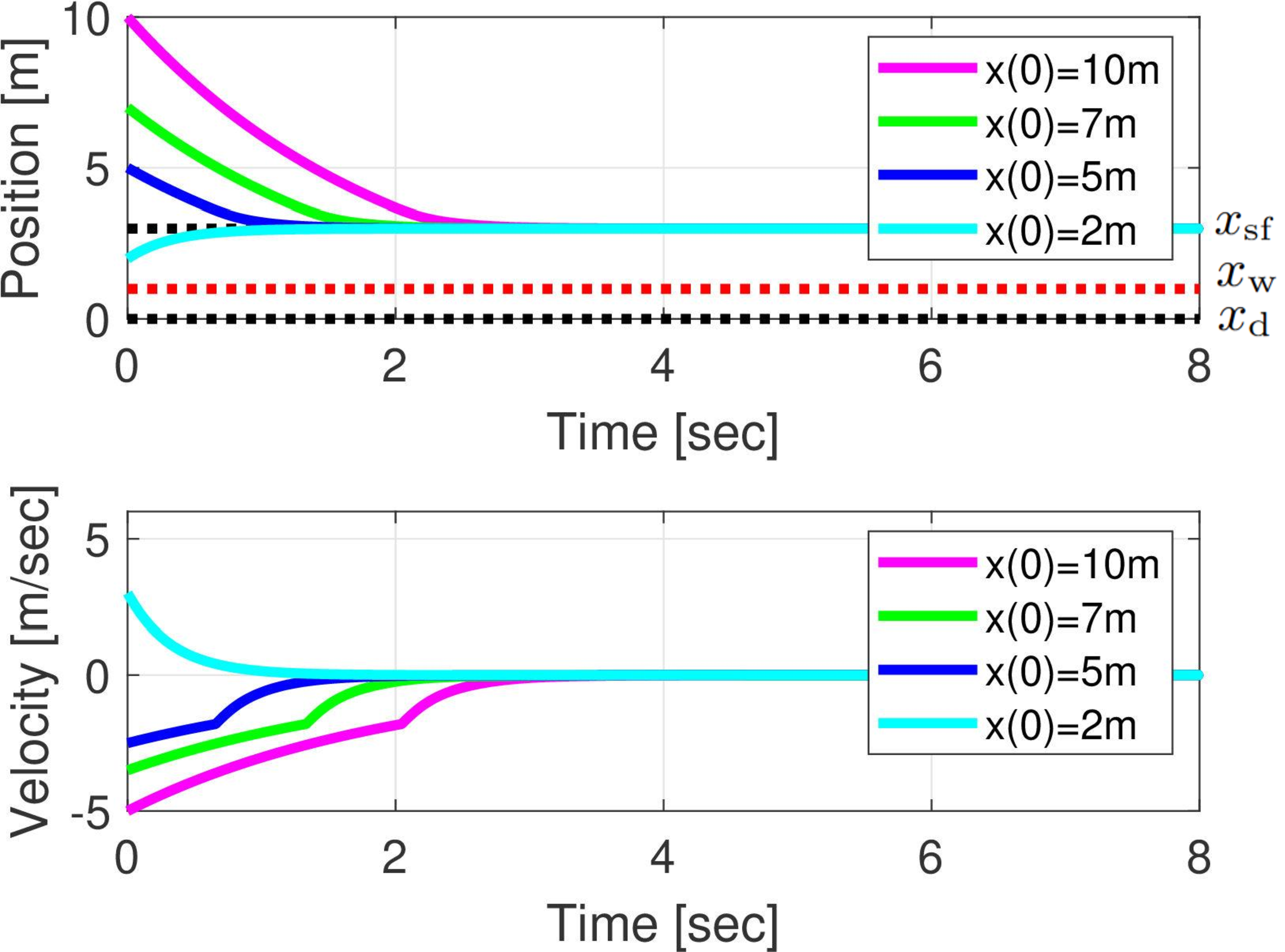}  % The printed column 
\caption{Position \& velocity for baseline zeroing CBF} % width is 8.4 cm.
\label{fig.1stOrderBenchmark1}                 % Size the figures 
\end{center}                 % accordingly.
\end{figure}
%%%%%%%%%%%%%%%%%%%%%%%%%%%%%%%%%%%%%%%%%%%%%%%%%%%%%

%%%%%%%%%%%%%%%%%%%%%%%%%%%%%%%%%%%%%%%%%%%%%%%%%%%%%
\begin{figure} 
\begin{center}
\includegraphics[width=0.45\textwidth]{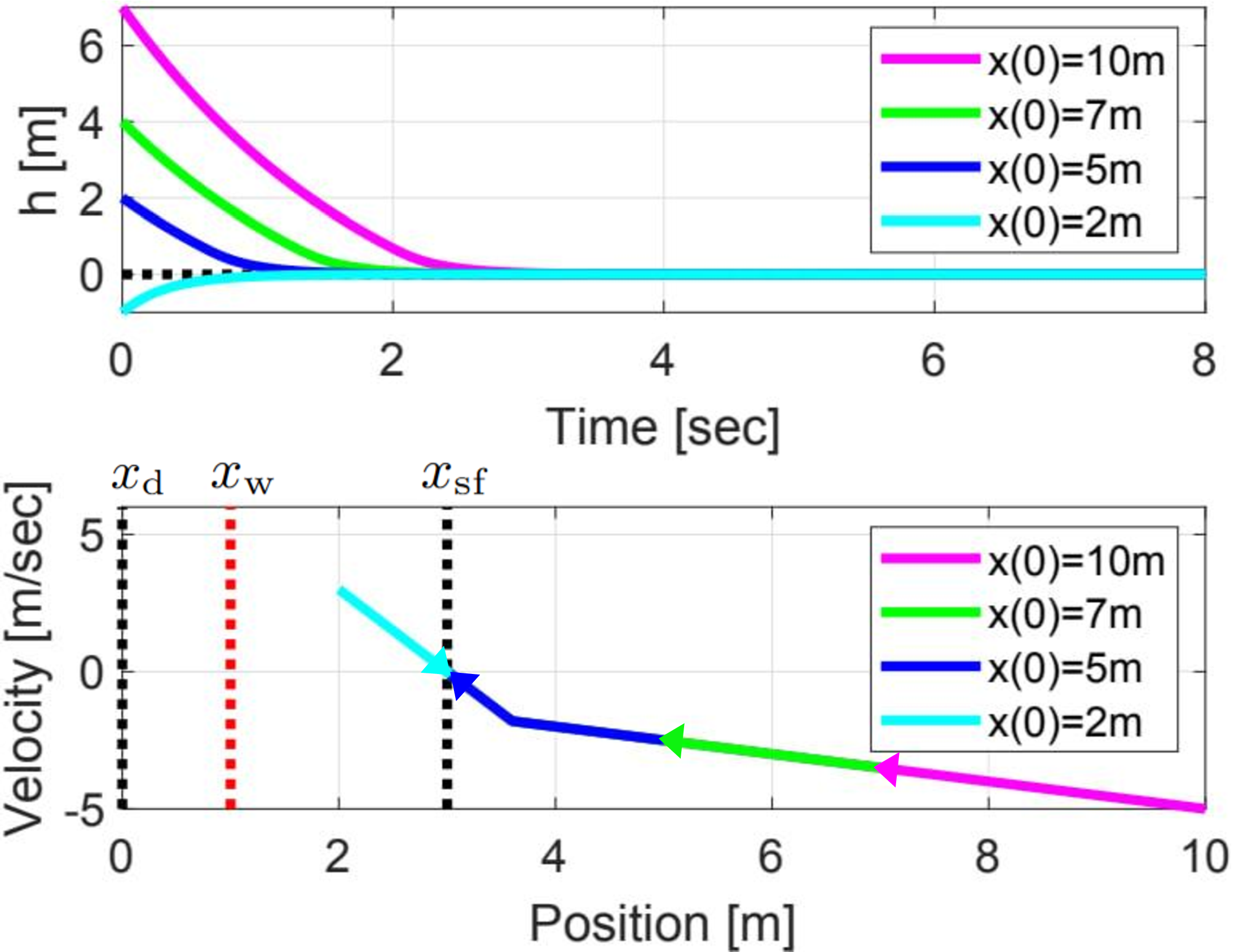}  % The printed column 
\caption{Value of the barrier function \& velocity vs. position plot for baseline zeroing CBF} % width is 8.4 cm.
\label{fig.1stOrderBenchmark2}                 % Size the figures 
\end{center}                 % accordingly.
\end{figure}
%%%%%%%%%%%%%%%%%%%%%%%%%%%%%%%%%%%%%%%%%%%%%%%%%%%%%

The above simulation results illustrate some of the key strengths of traditional zeroing CBFs. 
For example, the zeroing CBF succeeds in maintaining system safety as long as the initial object position is safe. 
Moreover, the zeroing CBF succeeds in bringing the object exponentially back towards safety if the initial condition is unsafe. 
This simulation study also highlights a well-recognized limitation of zeroing CBF methods. 
In particular, the fact that the safety control problem statement involves directly actuating the object's velocity, which is unrealistic in practical mechanical systems, where the actuation commands are typically forces and/or accelerations. 
This motivates the next simulation example, where acceleration is treated as the control input variable, leading to a relative degree of 2. 

%Although the object never collides with the wall in this simulation, the phase plane plot depicts a critical limitation of the benchmark approach -- its lack of "grace." The zeroing CBF creates the forward invariance guarantee around the safe distance requirement, $x_\mathrm{sf}$. However, this does not necessarily guarantee the avoidance of collision. Furthermore, this approach does not enable the controller to provide a stronger control input if the primary safe set breach is bigger.

%%%%%%%%%%%%%%%%%%%%%%%%%%%%%%%%%%%%%%%%%%%%%%%%%%%%%%%%%%%%%%
\subsection{Example \#2: a second-order exponential CBF}
%%%%%%%%%%%%%%%%%%%%%%%%%%%%%%%%%%%%%%%%%%%%%%%%%%%%%%%%%%%%%%

Suppose that the control input in the wall collision avoidance problem is acceleration, as opposed to velocity. 
This corresponds to the second-order differential equation
\begin{equation}\label{eq:2ndorderODE}
\ddot{x} = u,
\end{equation}
% or equivalently
% \begin{equation}\label{eq:2ndorderODE2}
% \begin{split}
% \dot{x} = v,
% \\
% \dot{v} = u.
% \end{split}
% \end{equation}
and indeed this system has two states, that is, ${\bx = ( x, \dot{x} )}$.
Again, for simplicity, omit the input bounds and assume ${u\in\mathbb{R}}$. 
Furthermore, consider the new baseline controller:  
\begin{equation} \label{eq.DesiredAcceleration} 
    u_\mathrm{d} = -k_1(x-x_\mathrm{d}) -k_2\dot{x} ,
%    \ddot{x}_\mathrm{d} = -k_1\dot{x}-k_2(x-x_\mathrm{d}),
\end{equation} 
where $k_1$ and $k_2$ are user-selected position and velocity feedback terms, respectively. 
Now consider the problem of ensuring the safety of the above object while utilizing the same barrier function \eqref{eq.1stCBF_h} as in the previous example.
Due to having a relative degree of 2 we build a high-order, specifically an exponential CBF \eqref{eq:reldeg2CBF_exp}
which results in
\begin{equation}\label{eq:h2_exp_example}
h_2(x,\dot{x}) = \dot{x} + \gamma_1 (x-x_\mathrm{sf}),
\end{equation}
yielding the set
\begin{equation}
 S_2 = \{ (x,\dot{x}) \in \mathbb{R}^2:\dot{x} \geq -\gamma_1 (x-x_\mathrm{sf}) \}.
\end{equation}
The safety constraint \eqref{eq:reldeg2CBF_exp}, \eqref{eq:reldeg2CBF_exp2} becomes
\begin{equation}\label{eq:reldeg2CBF_exp2_example}
u \geq  - \gamma_1 \gamma_2 (x - x_{\rm sf}) - (\gamma_1 + \gamma_2) \dot{x} =: u_\mathrm{sf},
\end{equation}
which ensures the forward invariance of the set
\begin{equation}\label{eq:invset_exp}
 S \cap S_2 = \{ (x,\dot{x}) \in \mathbb{R}^2: x \geq x_\mathrm{sf}, \dot{x} \geq -\gamma_1 (x-x_\mathrm{sf}) \}.
\end{equation}

In this case the QP \eqref{eq.1stOrderBenchmark_Optimization} can be formulated the same way,
keeping in mind that the control input $u$ now represents acceleration rather than velocity.
Now, solving the KKT condition yields the controller
\begin{equation}\label{eq.fromKKT2}
\begin{split}
u^* &= \max\{u_\mathrm{d},u_\mathrm{sf}\} 
\\
&= \max\{-k_1(x-x_\mathrm{d})-k_2\dot{x}, 
\\
&\qquad\qquad - \gamma_1 \gamma_2 (x - x_{\rm sf}) - (\gamma_1 + \gamma_2) \dot{x}\},
\end{split}
\end{equation}
with $u_\mathrm{d}$ and $u_\mathrm{sf}$ defined in \eqref{eq.DesiredAcceleration} and \eqref{eq:reldeg2CBF_exp2_example}. 
Substituting this into \eqref{eq:2ndorderODE} leads to the closed-loop system.

As in the previous example, the above safety controller is simulated in MATLAB using the ``ode15s" function. 
The same 8-second time horizon with a time step of 1 ms and the same parameter values for $x_\mathrm{d}$, $x_\mathrm{w}$, and $x_\mathrm{sf}$ are implemented as before, with additional parameter values listed in Table~\ref{tab:exp_parameters2}. 
The controller is simulated using four different initial position values of $2\ {\rm m}$, $5\ {\rm m}$, $7\ {\rm m}$, and $10\ {\rm m}$ while the initial object velocity is set as $-25\ {\rm m/s}$.
Now, these initial positions represent an extremely unsafe distance, two slightly unsafe distances, and a safe distance, respectively.

%%%%%%%%%%%%%%%%%%%%%%%%%%%%%%%%%%%%%%%%%%%%%%%%%%%%%
\begin{table}
\centering
    \caption{Additional Parameters used in Example \#2} \vspace{4pt} 
 \label{tab:exp_parameters2}
    \begin{tabular}{|c|c|c|}
    \hline
    Symbol & Definition & Value 
    \\
    \hline
    $k_1$ & Baseline proportional gain & $1$ ${\rm s}^{-2}$ 
    \\
    \hline
    $k_2$ & Baseline differential gain & $2$ ${\rm s}^{-1}$ 
    \\
    \hline
    $\gamma_1$ & 1st-order CBF gain & $4.5$ ${\rm s}^{-1}$ 
    \\
    \hline
    $\gamma_2$ & 2nd-order CBF gain & $0.5$ ${\rm s}^{-1}$ 
    \\
    \hline
    \end{tabular} 
\end{table}
%%%%%%%%%%%%%%%%%%%%%%%%%%%%%%%%%%%%%%%%%%%%%%%%%%%%%

Figure~\ref{fig.2ndOrder_eCBF1} plots the object's position, velocity, and acceleration over time, while Fig.~\ref{fig.2ndOrder_eCBF2} shows the value of the barrier functions and the phase plane plot. 
In the phase plane plot, the orange dotted line represents  ${h_2(x,\dot{x})=0}$, that is, ${\dot{x} = - \gamma_1 (x-x_\mathrm{sf})}$, cf.~\eqref{eq:h2_exp_example}. 
Lastly, once the object collides with the wall, the remainder of the corresponding simulation is plotted as a gray dashed curve instead of a solid colored curve. 
As seen in these figures, when ${x(0)=2\ {\rm m}}$ and ${x(0)=5\ {\rm m}}$, the controller fails to prevent wall collisions. 

%%%%%%%%%%%%%%%%%%%%%%%%%%%%%%%%%%%%%%%%%%%%%%%%%%%%%
\begin{figure}
\begin{center}
\includegraphics[width=0.59\textwidth, angle=270]{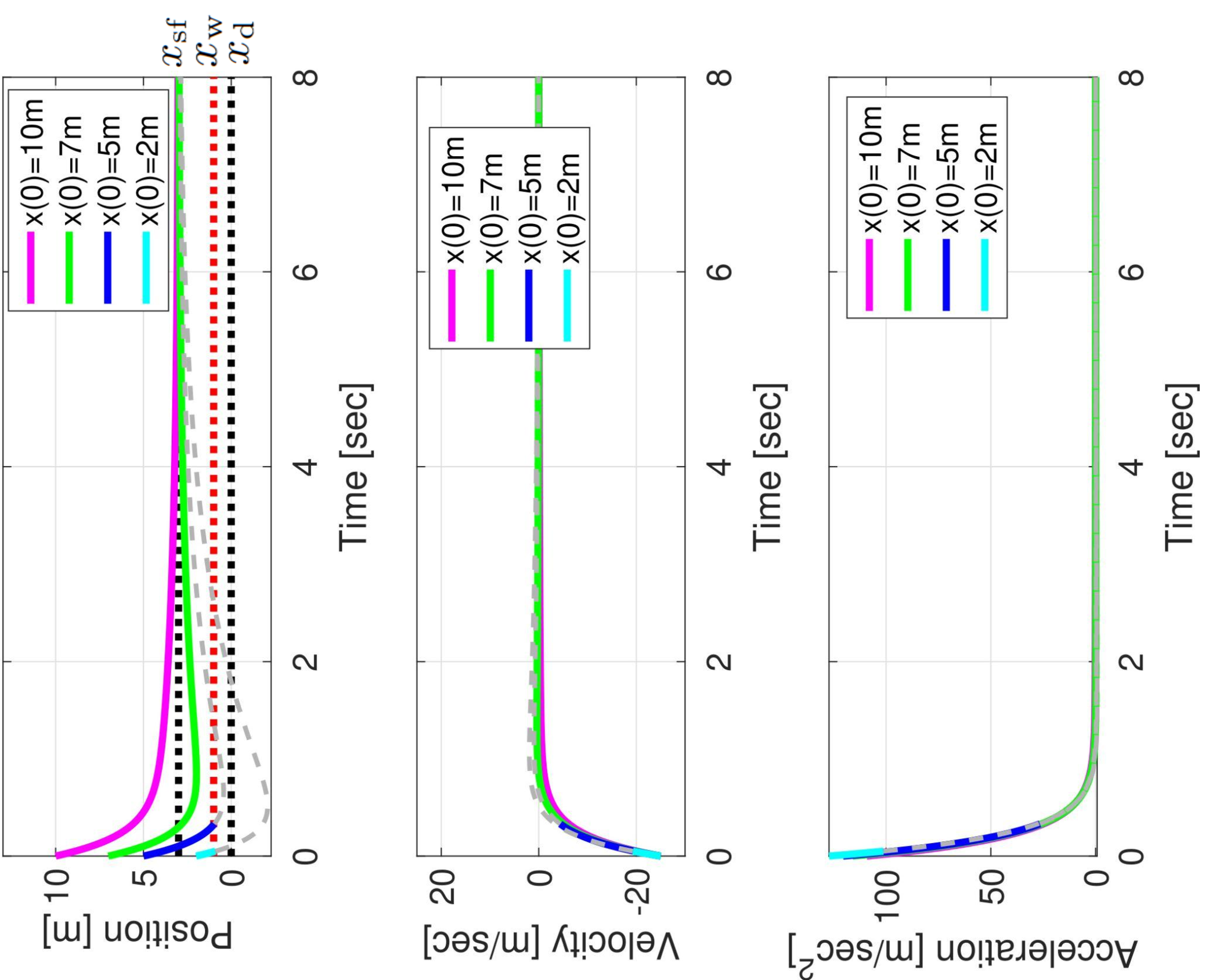}  % The printed column 
\caption{State \& input trajectories for the controller based on the exponential CBF} % width is 8.4 cm.
\label{fig.2ndOrder_eCBF1}                 % Size the figures 
\end{center}                 % accordingly.
\end{figure}
%%%%%%%%%%%%%%%%%%%%%%%%%%%%%%%%%%%%%%%%%%%%%%%%%%%%%

%%%%%%%%%%%%%%%%%%%%%%%%%%%%%%%%%%%%%%%%%%%%%%%%%%%%%
\begin{figure}
\begin{center}
\includegraphics[width=0.59\textwidth, angle=270]{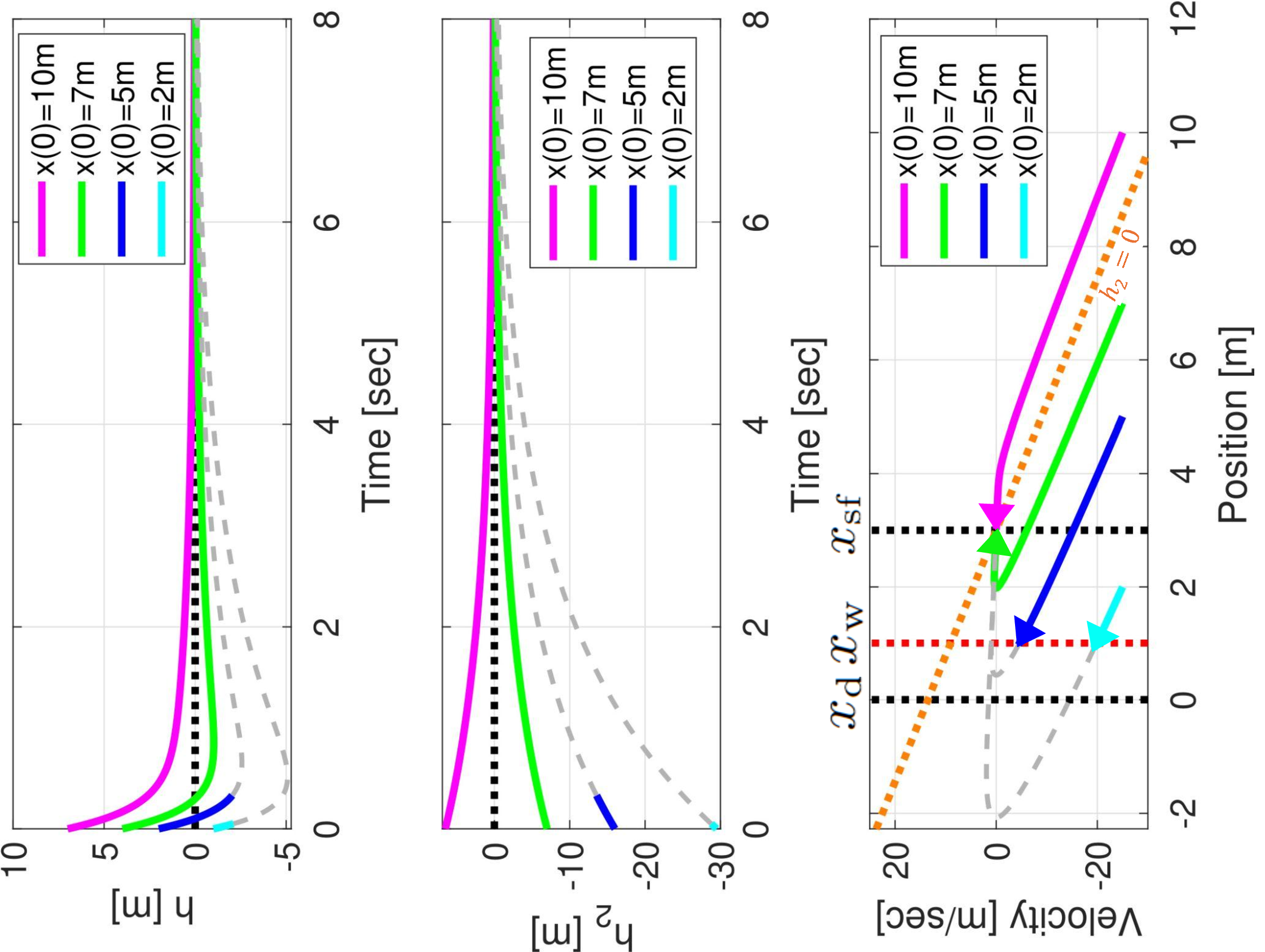}  % The printed column 
\caption{Barrier \& phase plot for the controller based on the exponential CBF} % width is 8.4 cm.
\label{fig.2ndOrder_eCBF2}                 % Size the figures 
\end{center}                 % accordingly.
\end{figure}
%%%%%%%%%%%%%%%%%%%%%%%%%%%%%%%%%%%%%%%%%%%%%%%%%%%%%

The above collisions illustrate one of the main weaknesses of traditional exponential safety barriers. 
This is not a consequence of actuation constraints. 
In fact, while such constraints are relevant in practical applications, they are deliberately left out of this section's examples for illustrative reasons. 
Rather, the culprit for this failure to prevent collisions is more fundamental. 
From a mathematical perspective, the above exponential safety controller only guarantees the forward invariance of the set ${S \cap S_2}$ given in \eqref{eq:invset_exp}. 
Intuitively, this translates to a position- and velocity-dependent definition of what this paper describes as the ``primary" safety boundary. 
When the initial velocity of the object is too large in magnitude and negative, even if the initial position of the object is ``safe" (based on the definition of $h(\bx)$), the ``primary" safety boundary is breached, and the exponential safety controller is no longer guaranteed to prevent collisions with the wall. 
This motivates the focus of this paper on adding one more, ``failsafe" safety boundary. 
The core idea behind graceful safety control, as described in the next section, is to ensure that this failsafe boundary continues to be honored at all moments in time, even if the primary safety boundary is breached.

%%%%%%%%%%%%%%%%%%%%%%%%%%%%%%%%%%%%%%%%%%%%%%%%%%%%%%%%%%%%%%
\section{Graceful safety control}\label{sec:grace}
%%%%%%%%%%%%%%%%%%%%%%%%%%%%%%%%%%%%%%%%%%%%%%%%%%%%%%%%%%%%%%

This section proposes a first- and second-order graceful safety control policies. 
We begin by presenting the first-order policy and proving its ability to achieve graceful control. 
Then we extend this policy to second-order, which allows the application of this framework to systems of relative degree 2.

%%%%%%%%%%%%%%%%%%%%%%%%%%%%%%%%%%%%%%%%%%%%%%%%%%%%%%%%%%%%%%
 \subsection{First-order graceful safety control framework}
 %%%%%%%%%%%%%%%%%%%%%%%%%%%%%%%%%%%%%%%%%%%%%%%%%%%%%%%%%%%%%%

Consider a dynamic system governed by \eqref{eq.ControlAffineSys}. 
Suppose that this system has a relative degree of 1. 
Moreover, suppose that there exists a continuously differentiable function ${H:\mathbb{R}^n\rightarrow\mathbb{R}}$ that can be used for defining system safety. 
In particular, suppose that there exist two boundaries associated with $H(\bx)$: a primary boundary at ${H(\bx)=b}$ below which the system is in ``danger", and a secondary boundary at ${H(\bx)=a}$, where ${a<b}$, below which the system experiences a safety ``catastrophe". 
The specific values of the safety thresholds, $a$ and $b$, are problem-dependent. 
However, one can construct an affine transformation to shift these boundaries as 
\begin{equation}\label{eq:h_g}
    h_g(\bx) = \frac{H(\bx)-b}{b-a}.
\end{equation}
Based on the above transformation, a given system is ``safe" when ${h_g(\bx) \geq 0}$, in ``danger" when ${-1<h_g(\bx)<0}$, and ``catastrophic" when ${h_g(\bx) \leq -1}$. 
In other words, the above definitions of $H(\bx)$ and $h_g(\bx)$ correspond to a multi-layered definition of safety: an essential prerequisite to graceful safety control. 
This multi-layered definition is stated formally below. 
\\

\begin{definition}
Consider the system in \eqref{eq.ControlAffineSys} and the continuously differentiable function ${h_g(\bx):\mathbb{R}^n \rightarrow\mathbb{R}}$, then the system's primary safe set $S$, secondary failsafe set $D$, and catastrophically unsafe set $C$ are defined as 
\begin{equation}\label{eq.gracefulsafetyset}
\begin{split}
S &=\{\bx \in\mathbb{R}^n : h_g(\bx) \geq 0\},  
\\
D &=\{\bx \in\mathbb{R}^n : -1 < h_g(\bx) < 0 \},
\\
C &=\{\bx \in\mathbb{R}^n : h_g(\bx) \leq -1 \}.
\end{split}
\end{equation}
\end{definition}

We propose to achieve graceful safety control by imposing the constraint  
\begin{equation} \label{eq.1stgraceful}
  \dot{h}_g(\bx,\bu)  \geq -\beta\left(\frac{h_g(\bx)}{h_g(\bx)+1}\right),
\end{equation}
at every instant in time, where $\beta$ is an extended class-$\mathcal{K}$ function.
For simplicity, we will choose ${\beta(r) = \gamma r}$ in the examples detailed below.
Fig.~\ref{Fig.1stGracefulCBF} visualizes the right-hand side of the constraint \eqref{eq.1stgraceful} for ${\beta(r) = r}$. 

The constraint \eqref{eq.1stgraceful} achieves two goals simultaneously.
First, it functions as a zeroing CBF around ${h_g=0}$.
Second, it provides a stiffening response to increasing safety violations that becomes infinitely large in the limit as $h_g$ approaches $-1$. 
Together, these two properties provide a multi-layered graceful safety guarantee, as shown below. 

\textbf{Remark:} While the above general setup does not require it, it is reasonable to construct $h_g(\bx)$ such that ${\nabla h_g(\bx) < \mathbf{0}}$ for ${\bx \in \partial S \cup D \cup \partial C}$, cf.~\eqref{eq.gracefulsafetyset}.
This ensures as the system goes from the primary safety boundary $\partial S$, given by ${h_g(\bx) = 0}$, toward the boundary of catastrophe $\partial C$, given by ${h_g(\bx) = -1}$, within the failsafe set $D$, there is a ``sense of increasing danger". 

%%%%%%%%%%%%%%%%%%%%%%%%%%%%%%%%%%%%%%%%%%%%%%%%%%%%%
\begin{figure}[!t]
\begin{center}
\includegraphics[width=0.47\textwidth]{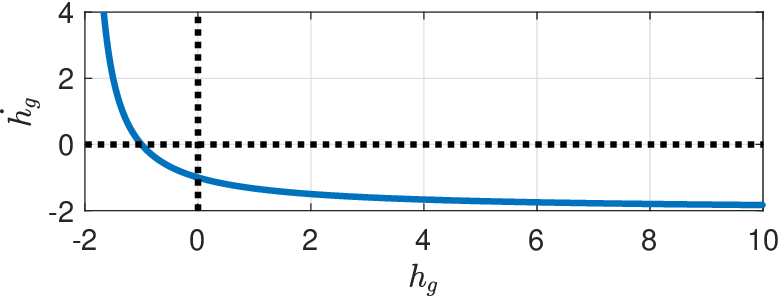}  % The printed column 
\caption{Illustrating the inequality \eqref{eq.1stgraceful} for ${\beta(r) = r}$.} % width is 8.4 cm.
\label{Fig.1stGracefulCBF}                 % Size the figures 
\end{center}                 % accordingly.
\end{figure}
%%%%%%%%%%%%%%%%%%%%%%%%%%%%%%%%%%%%%%%%%%%%%%%%%%%%%

%%%%%%%%%%%%%%%%%%%%%%%%%%%%%%%%%%%%%%%%%%%%%%%%%%%%%%%%%%%%%%
\subsection{First-order graceful safety guarantees}
%%%%%%%%%%%%%%%%%%%%%%%%%%%%%%%%%%%%%%%%%%%%%%%%%%%%%%%%%%%%%%

The graceful safety controller in \eqref{eq.1stgraceful} has three important properties, as stated in the following theorems.
\\

\theoremstyle{theorem}
\begin{theorem}
Consider a continuously differentiable function $h_g(\bx)$ constructed as in \eqref{eq:h_g} for the control system \eqref{eq.ControlAffineSys}. 
Then, any locally Lipschitz continuous controller $\bu(\bx)$ that satisfies \eqref{eq.1stgraceful} guarantees the forward invariance of $S$.  
\end{theorem}
\vspace{-12pt}
\begin{proof}
The proof follows directly from the fact that $h_g$ satisfies the definition of a zeroing CBF. 
In particular, the right-hand side of \eqref{eq.1stgraceful} is an extended class-$\mathcal{K}$ function of $h_g$, since $\beta$ and
$r/(r+1)$ are extended class-$\mathcal{K}$ functions and the composition of them is an extended class-$\mathcal{K}$ function.
\end{proof}

%%%%%%%%%%%%%%%%%%%%%%%%%%%%%%%%%%%%%%%%%%%%%%%%%%%%%
\begin{figure}[!t]
\begin{center}
\includegraphics[width=0.47\textwidth]{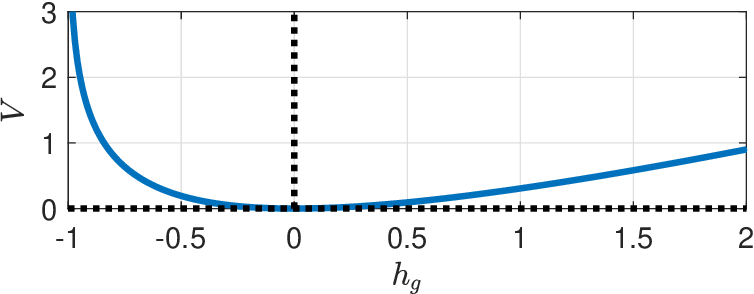}  % The printed column 
\caption{Lyapunov function candidate \eqref{eq:Lyap}} % width is 8.4 cm.
\label{Fig.1stOrderLyapunov}                 % Size the figures 
\end{center}                 % accordingly.
\end{figure}
%%%%%%%%%%%%%%%%%%%%%%%%%%%%%%%%%%%%%%%%%%%%%%%%%%%%%

\textbf{Remark:} Intuitively, the above theorem asserts that the proposed graceful controller inherits the properties of traditional zeroing CBFs around the primary safety boundary $\partial S$. 
In other words, the primary safe set $S$ is guaranteed to be forward invariant. 
\\

\begin{theorem}
Suppose that \eqref{eq.1stgraceful} holds and ${\bx(0) \in D}$,  i.e.,  ${-1<h_g(\bx(0))<0}$. Then  ${\bx(t) \in S \cup D}$, that is, ${h_g(\bx(t))>-1}$ for all ${t\geq0}$, and $\bx(t)$ converges to $\partial S$, that is, ${\lim_{t\to\infty}h_g(\bx(t))=0}$. 
\end{theorem}
\vspace{-12pt}
\begin{proof}
Consider the Lyapunov candidate function 
\begin{equation}\label{eq:Lyap}
V(h_g)=h_g-\ln(h_g+1),
\end{equation}
depicted in Fig.~\ref{Fig.1stOrderLyapunov}.
This function equals zero if and only if ${h_g=0}$, and is otherwise positive. 
Moreover, the derivative of this Lyapunov candidate function is given by
\begin{equation}\label{eq:lider1} 
    \dot{V}=\dot{h}_g-\frac{\dot{h}_g}{|h_g+1|}.
\end{equation}
For ${h_g>-1}$ we have ${|h_g+1|= h_g+1}$, which results in
\begin{equation}\label{eq.LyapunovProof}
    \dot{V}=\dot{h}_g\left(1-\frac{1}{h_g+1}\right) = \dot{h}_g\frac{h_g}{h_g+1}.
\end{equation}
Multiplying both sides of the inequality \eqref{eq.1stgraceful} with ${h_g/(h_g+1)<0}$ yields
\begin{equation}\label{eq:lider2} 
\dot{h}_g\frac{h_g}{h_g+1} \leq -
\beta \left(\frac{h_g}{h_g+1}\right)  \frac{h_g}{h_g+1}.
\end{equation}
Comparing \eqref{eq:lider1} and \eqref{eq:lider2} results in
\begin{equation}\label{eq:lider3} 
    \dot{V} \leq -\eta \left(\frac{h_g}{h_g+1}\right),
\end{equation}
where $\eta$ is a class-$\mathcal{K}$ function.
Therefore, exploiting that ${-1<h_g<0}$, we obtain ${\dot{V}<0}$. 
Since the time derivative of the Lyapunov candidate function is negative at all points where the function is nonzero implies that ${h_g=0}$ is asymptotically stable.
This implies that ${S\cup D}$ is forward invariant and that the systems converges to $\partial S$ as $t\to\infty$, which completes the proof.
\end{proof}

\textbf{Remark:} Intuitively, the above theorem asserts that as long as the system is initialized in the failsafe set $D$, it is guaranteed to avoid catastrophic failure. 
This is a consequence of the ``stiffening" feedback effect provided by the nonlinear safety constraint \eqref{eq.1stgraceful}. 
In fact, one can intuitively think of this stiffening effect as being analogous to a reciprocal CBF around ${h_g=-1}$.  

\textbf{Remark:} The above choice of Lyapunov candidate function is inspired by the fact that the nonlinear term in the proposed safety constraint acts in a manner analogous to a stiffening actuator (e.g., a stiffening spring) in the limit as ${h_g\rightarrow -1}$. 
In fact, the Lyapunov candidate function \eqref{eq:Lyap} was constructed by analytically integrating the area under the function ${h_g/(h_g+1)}$ starting with ${h_g=0}$, in a manner analogous to computing the energy stored in a stiffening spring as a function of deflection. 
This analogy is particularly valuable for the extension of the proposed graceful controller to systems of relative degree 2, as discussed further below. 
\\

% \begin{theorem}
% Suppose that ${-1<h_g(\bx(t))<0}$ over a time interval ${t\in[t_\mathrm{o},t_\mathrm{f}]}$. 
% Then $h_g(\bx(t))$ approaches zero monotonically throughout this interval. 
% \end{theorem}

% \begin{proof}
% The above Lyapunov stability analysis shows that $V(h_g(\bx(t)))$ diminishes monotonically at all times during the above time interval. 
% Given the monotonicity of the Lyapunov candidate function \eqref{eq:Lyap} with respect to $h_g$ for all values on the open set ${(0,1)}$, it follows that $h_g$ monotonically approaches zero over the time interval ${t\in[t_\mathrm{o},t_\mathrm{f}]}$.
% \end{proof}

%%%%%%%%%%%%%%%%%%%%%%%%%%%%%%%%%%%%%%%%%%%%%%%%%%%%%
\begin{figure}[!t]
\begin{center}
\includegraphics[width=0.47\textwidth]{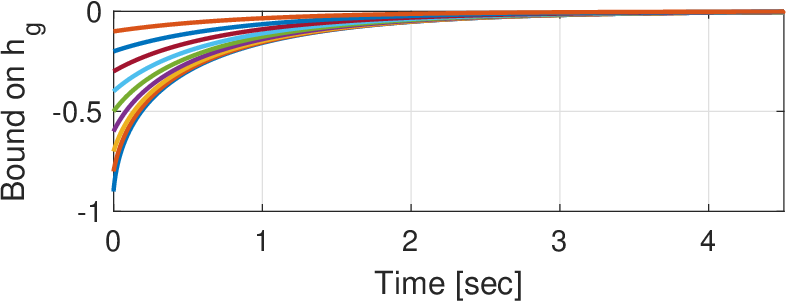}  % The printed column 
\caption{Evolution of the bound on $h_g$ over time} % width is 8.4 cm.
\label{Fig.bound}                 % Size the figures 
\end{center}                 % accordingly.
\end{figure}
%%%%%%%%%%%%%%%%%%%%%%%%%%%%%%%%%%%%%%%%%%%%%%%%%%%%%

\textbf{Remark:} To gain further insight into the above monotonic convergence towards safety, one may consider equality and the linear function ${\beta(r)=\gamma r}$ in \eqref{eq.1stgraceful}:
\begin{equation}  
\frac{{\rm d}h_g}{{\rm d}t}=-\gamma \frac{h_g}{h_g+1}.
\end{equation}
This can be solved using the separation of variables:
\begin{equation} \label{eq.ImplicitSolution}
\begin{split}
  &\int_{h_g(0)}^{h_g(\tau)} \frac{h_g+1}{h_g} {\rm d}h_g = -\gamma\int_{0}^{\tau} {\rm d}t
  \\
  & \tau = -\frac{1}{\gamma}\big(h_g(\tau)+ \ln|h_g(\tau)|-h_g(0) - \ln|h_g(0)|\big).
  \end{split}
\end{equation}
The function $h_g(\tau)+ \ln|h_g(\tau)|$ is monotonic in terms of $\tau$. 
Therefore, according to the comparison lemma, the above implicit equation provides a monotonic bound on the evolution of $h_g(\tau)$ versus time.
This bound reinforces the insight that the proposed graceful safety controller returns monotonically towards the primary safety boundary if it is breached. 
Fig.~\ref{Fig.bound} illustrates this behavior by plotting the bound on $h_g$ given by \eqref{eq.ImplicitSolution} as a function of time for different initial violations of the primary safety condition. 
Graceful safety control, in this case, imposes monotonic bounds on convergence towards the primary safe set.
If ${\nabla h_g(\bx) < \mathbf{0}}$ for ${\bx \in \partial S \cup D \cup \partial C}$ also holds, this monotonicity also results in a "streamlined" flow in state space.

%%%%%%%%%%%%%%%%%%%%%%%%%%%%%%%%%%%%%%%%%%%%%%%%%%%%%%%%%%%%%%
\subsection{Second-order graceful safety control framework}
%%%%%%%%%%%%%%%%%%%%%%%%%%%%%%%%%%%%%%%%%%%%%%%%%%%%%%%%%%%%%%

This section extends the core idea behind graceful safety control to systems of relative degree 2. 
Specifically, we consider the second-order exponential barrier constraint in~\eqref{eq:reldeg2CBF_exp2} as a motivation.
Suppose that the lowest-order term in this constraint is replaced with a stiffening term, analogous to~\eqref{eq.1stgraceful} with ${\beta(r)=\gamma_1 r}$, to furnish the following nonlinear second-order constraint 
\begin{equation}\label{eq.2ndGraceful0}
\ddot{h}_g + (\gamma_1 + \gamma_2 ) \dot{h}_g + \gamma_1 \gamma_2 \frac{h_g}{h_g+1} \geq 0, 
\end{equation}
Alternatively, one can express this equation in a more general form: 
\begin{equation} \label{eq.2ndGraceful}
%\begin{split}
%&\ddot{h}_g+2\zeta\omega\dot{h}_g+\omega^2(1-\frac{1}{h_g+1})\geq 0 
%\\
%\implies &
\ddot{h}_g + 2\zeta\omega_\mathrm{n} \dot{h}_g + \omega_\mathrm{n}^2 \frac{h_g}{h_g+1}\geq 0. 
%\end{split}
\end{equation}
The structure of the above constraint is analogous to a nonlinear mass-spring-damper system where the term $h_g/(h_g+1)$ provides a stiffening effect. 
Moreover, when this inequality is linearized around ${h_g(x)=0}$, it becomes analogous to a standard mass-spring damper with a natural frequency $\omega_\mathrm{n}$ and a damping ratio $\zeta$. 
One can even make $\gamma_1$ and $\gamma_2$ a complex conjugate pair (with negative real parts) which corresponds to ${0<\zeta<1}$, i.e., an underdamped system.

In general, we consider that $\dot{h}_g(\bx,\bu)$ does not depend on $\bu$, i.e., we have $\dot{h}_g(\bx)$. 
Then we impose the constraint
\begin{equation} \label{eq.2ndgraceful_general}
  \dot{h}_g(\bx,\bu)  \geq - 2\zeta\omega_\mathrm{n} \dot{h}_g(\bx) - \omega_\mathrm{n}^2 \frac{h_g(\bx)}{h_g(\bx)+1}.
\end{equation}
with ${\zeta>0}$ and $\omega_\mathrm{n}>0$ to guarantee graceful safety.

%%%%%%%%%%%%%%%%%%%%%%%%%%%%%%%%%%%%%%%%%%%%%%%%%%%%%%%%%%%%%%
\subsection{Second-order graceful safety guarantees}
%%%%%%%%%%%%%%%%%%%%%%%%%%%%%%%%%%%%%%%%%%%%%%%%%%%%%%%%%%%%%%

The above second-order inequality provides two different safety guarantees, as shown below. 
\\

\begin{theorem}
Suppose that \eqref{eq.2ndgraceful_general} holds and ${\bx(0) \in D}$,  i.e.,  ${-1<h_g(\bx(0))<0}$. 
Then  ${\bx(t) \in S \cup D}$, that is, ${h_g(\bx(t))>-1}$ for all ${t\geq0}$, and $\bx(t)$ converges to $\partial S$, that is, ${\lim_{t\to\infty}h_g(\bx(t))=0}$.
% Consider the following set of unsafe states: 
% \begin{equation}
%     S_u := \{x:-1<h_g(x)<0\}
% \end{equation}
% For all initial states in this set, regardless of the initial value of $\dot{h}_g$, the system governed by Eq.~\ref{eq.2ndGraceful} is guaranteed to converge asymptotically to $h_g=0$. 
\end{theorem}
\vspace{-12pt}
\begin{proof}
We begin by constructing the following Lyapunov candidate function
\begin{equation}\label{eq.Lyap2}
    V(h_g,\dot{h}_g) = \frac{1}{2} \dot{h}_g^2 \Big( 1 - \Theta\big(\dot{h}_g\big) \Big) + \omega_\mathrm{n}^2 \big(h_g-\ln(h_g+1) \big),
\end{equation}
where $\Theta$ denotes the unit step function, that is,
\begin{equation}
\begin{split}
\Theta(x)=
\begin{cases}
1, &\!\!\!\text{if}\,\,\, x\geq 0,
\\
0, &\!\!\!\text{if}\,\,\, x < 0,
\end{cases}
\,\,\,\Rightarrow\,\,\,
1-\Theta(x)=
\begin{cases}
0, &\!\!\!\text{if}\,\,\, x\geq 0,
\\
1, &\!\!\!\text{if}\,\,\, x < 0.
\end{cases}
\end{split}
\end{equation}

This Lyapunov candidate function is continuously differentiable, since the first derivative of this function with respect to $\dot{h}_g$ exists everywhere (while the second derivative does not exist at ${\dot{h}_g=0}$). 
Moreover, ${V=0}$ when ${h_g=0}$ and ${\dot{h}_g\geq0}$ and ${V>0}$ otherwise. 
The structure of this Lyapunov candidate function is analogous to a potential energy term plus a kinetic energy term that is computed only when ${\dot{h}_g<0}$. 

When $\dot{h}_g>0$, the time derivative of \eqref{eq.Lyap2} this candidate function is given by 
\begin{equation}
    \dot{V} = \omega_\mathrm{n}^2 \left(1-\frac{1}{|h_g+1|}\right) \dot{h}_g
    = \omega_\mathrm{n}^2 \frac{h_g}{h_g+1} \dot{h}_g,
\end{equation}
where the last equality holds since ${-1<h_g<0}$. 
This derivative is guaranteed to be negative. 
Moreover, when ${\dot{h}_g<0}$, the time derivative of the Lyapunov candidate function becomes
\begin{equation}
    \dot{V} = \left(\ddot{h}_g + \omega_\mathrm{n}^2 \frac{h_g}{h_g+1}\right) \dot{h}_g,
\end{equation}
where, again, we exploited ${-1<h_g<0}$. 
According to \eqref{eq.2ndGraceful}, the expression within the parenthesis is larger than ${-2\zeta\omega_\mathrm{n}\dot{h}_g}$
% \begin{equation} 
% \ddot{h}_g + \omega_\mathrm{n}^2 \frac{h_g}{h_g+1} \geq -2\zeta\omega_\mathrm{n}\dot{h}_g
% \end{equation}
Since $\dot{h}_g$ is negative, it follows that ${\dot{V}<0}$ and according to Lyapunov's theorem ${h_g=0}$ is asymptotically stable. 
This implies that ${S\cup D}$ is forward invariant and that the systems converges to $\partial S$ as $t\to\infty$, which completes the proof.
\end{proof}

\vspace{-12pt}
\textbf{Remark:} The above result is essentially a failsafe guarantee associated with the proposed controller. 
It implies that even if the controller's primary safety threshold is breached, the controller guarantees that the failsafe threshold at ${h_g=-1}$ is not be breached, regardless of $\dot{h}_g$. 
This result is a consequence of the stiffening behavior of the feedback term ${h_g/(h_g+1)}$ as $h_g$ approaches $-1$, and is essential to this paper's notion of graceful safety-critical control. 
Interestingly, this graceful guarantee does not require ${\zeta>1}$, meaning that it exists even if the control design parameter $\zeta$ provides underdamped linearized behavior around ${h_g=0}$. 

The following theorem shows that for ${0<\zeta<1}$ the forward invariance of the primary safety set is guaranteed.

\vspace{12pt}
\begin{theorem}
    Suppose that \eqref{eq.2ndgraceful_general} holds and suppose that ${\zeta>1}$ and ${\omega_\mathrm{n}>0}$. 
    Then the two sets
    \begin{equation} \label{eq.SafeSet2spec}
    \begin{split}
    S_{g,1} &= \{ \bx \in\mathbb{R}^n : h_g(\bx)) \geq 0, \dot{h}_g(\bx) + \gamma_1 h_g(\bx)) \geq 0\}, 
    \\
    S_{g,2} &= \{ \bx \in\mathbb{R}^n : h_g(\bx)) \geq 0, \dot{h}_g(\bx) + \gamma_2 h_g(\bx)) \geq 0\}, 
    \end{split}
    \end{equation}
    are both forward invariant, where $\gamma_1$ and $\gamma_2$ are the real solutions of the characteristic equation 
    \begin{equation}\label{eq.char}
    \lambda^2 + 2\zeta\omega_\mathrm{n}\lambda + \omega_\mathrm{n}^2=0.
    \end{equation}
\end{theorem}
\vspace{-12pt}
\begin{proof}
If ${\zeta>1}$ and $\omega_\mathrm{n}>0$, then the characteristic equation \eqref{eq.char} has two negative real eigenvalues. 
Denoting these eigenvalues by $-\gamma_1$ and $-\gamma_2$, one can write \eqref{eq.2ndgraceful_general} (or equivalently \eqref{eq.2ndGraceful}) into the form \eqref{eq.2ndGraceful0}.
% \begin{equation}
% \ddot{h}_g+(\lambda_1+\lambda_2)\dot{h}_g+\lambda_1\lambda_2(\frac{h_g}{h_g+1})\geq 0
% \end{equation}

If ${\bx \in S_{g,1}}$ or ${\bx \in S_{g,2}}$, based on the definitions in \eqref{eq.SafeSet2spec}, we have ${h_g\geq0}$.
%This implies that ${h_g+1\geq1}$, which in turn implies that $1/(h_g(x)+1)\leq1$. 
This implies that ${h_g\geq h_g/(h_g+1)}$, leading to the  inequality
\begin{equation}\label{eq:exponentialGrace}
\ddot{h}_g + (\gamma_1 + \gamma_2) \dot{h}_g + \gamma_1\gamma_2 h_g \geq 0.
\end{equation}
Now define the following two functions 
\begin{equation}
\begin{split}
    h_{g,1} &:= \dot{h}_g + \gamma_1 h_g,
\\
    h_{g.2} &:= \dot{h}_g + \gamma_2 h_g.
\end{split}    
\end{equation}
Plugging these into \eqref{eq:exponentialGrace} yields
\begin{equation}
\begin{split}
    \dot{h}_{g,1} + \gamma_1 h_{g,1} \geq 0,
\\
    \dot{h}_{g,2} + \gamma_2 h_{g,2} \geq 0.
\end{split}    
\end{equation}
These inequalities guarantee the forward invariance of the sets $S_{g,1}$ and $S_{g,2}$ defined in \eqref{eq.SafeSet2spec}.
This completes the proof.
\end{proof}
\vspace{-12pt}
\textbf{Remark:} Intuitively, the above theorem provides a conservative estimate of two forward invariant sets for the proposed controller.
The forward invariance of these two sets is analogous to the safety guarantees associated with second-order exponential CBFs. 
Also, one can intuitively interpret this result a velocity-dependent safe distance guarantee. 
This serves as the primary safety guarantee for the proposed controller. 
\vspace{12pt}

%\begin{proof}
%Consider a continuously differentiable CBF \eqref{eq.2ndGraceful}. Then, one can construct a candidate Lyapunov function, $V=\frac{\dot{h}_g^2(x)}{2}+\omega^2 (h_g(x)-ln(h_g+1))$. Notice that this Lyapunov function is positive everywhere except at $h_g(x)=0$ and $\dot{h}_g(x)=0$. The derivative of this Lyapunov function is $\dot{V}=-2\zeta\omega\dot{h}_g^2(x)$. This equation is negative everywhere except at $\dot{h}_g(x)=0$. Therefore, due to Lyapunov stability theorem, the CBF (\ref{eq.2ndGraceful}) is guaranteed to monotonically converge to its equilibrium point. Furthermore, since the only equilibrium point and the largest invariant set is $(\dot{h}_g(x),h_g(x))=(0,0)$, due to La Salle's invariance principle, it is guaranteed to converge asymptotically to its equilibrium, $h_g(x) = 0$. 
%\end{proof}

%\textit{Can we guarantee that 2nd order gCBF will always be a valid CBF?}

%%%%%%%%%%%%%%%%%%%%%%%%%%%%%%%%%%%%%%%%%%%%%%%%%%%%%%%%%%%%%%
\section{Simulation Study - Graceful Safety Control}\label{sec:simulation}
%%%%%%%%%%%%%%%%%%%%%%%%%%%%%%%%%%%%%%%%%%%%%%%%%%%%%%%%%%%%%%

In this section, a simulation study is conducted using the same example as in Sec.~\ref{sec:examples}, see Fig.~\ref{fig.object}, for both the first- and second-order graceful safety controllers. 
The simulation results clearly demonstrate that both of the proposed safety controllers are able to achieve ``grace" in the system by successfully preventing the object from colliding with the wall, unlike the second-order motivating example.

%%%%%%%%%%%%%%%%%%%%%%%%%%%%%%%%%%%%%%%%%%%%%%%%%%%%%%%%%%%%%%
\subsection{Scenario \#1: first-order graceful safety controller}
%%%%%%%%%%%%%%%%%%%%%%%%%%%%%%%%%%%%%%%%%%%%%%%%%%%%%%%%%%%%%%

We first construct the first-order graceful safety controller. We begin by defining the following barrier function:
\begin{equation} \label{eq.h_g}
  h_g(x) = \frac{x-x_\mathrm{sf}}{x_\mathrm{sf}-x_\mathrm{w}}.
\end{equation}
This barrier function allows the creation of three different safety sets, per \eqref{eq.gracefulsafetyset}. 
If the object is placed at ${x \geq x_\mathrm{sf}}$ we have ${h_g(x) \geq 0}$. 
This serves as the primary safe layer that enforces the desirable safety condition. 
If the object's position is ${x_\mathrm{w} < x <x_\mathrm{sf}}$ then we have ${-1 < h_g(x) < 0}$ and the system is considered to be in danger, but not yet catastrophic. 
When the object collides with the wall, i.e., ${x \leq x_\mathrm{w}}$ then ${h_g(x)\leq-1}$. 
This represents that the secondary failsafe layer is breached, resulting in a catastrophe.

%%%%%%%%%%%%%%%%%%%%%%%%%%%%%%%%%%%%%%%%%%%%%%%%%%%%%
\begin{figure} [t]
\begin{center}
\includegraphics[width=0.45\textwidth]{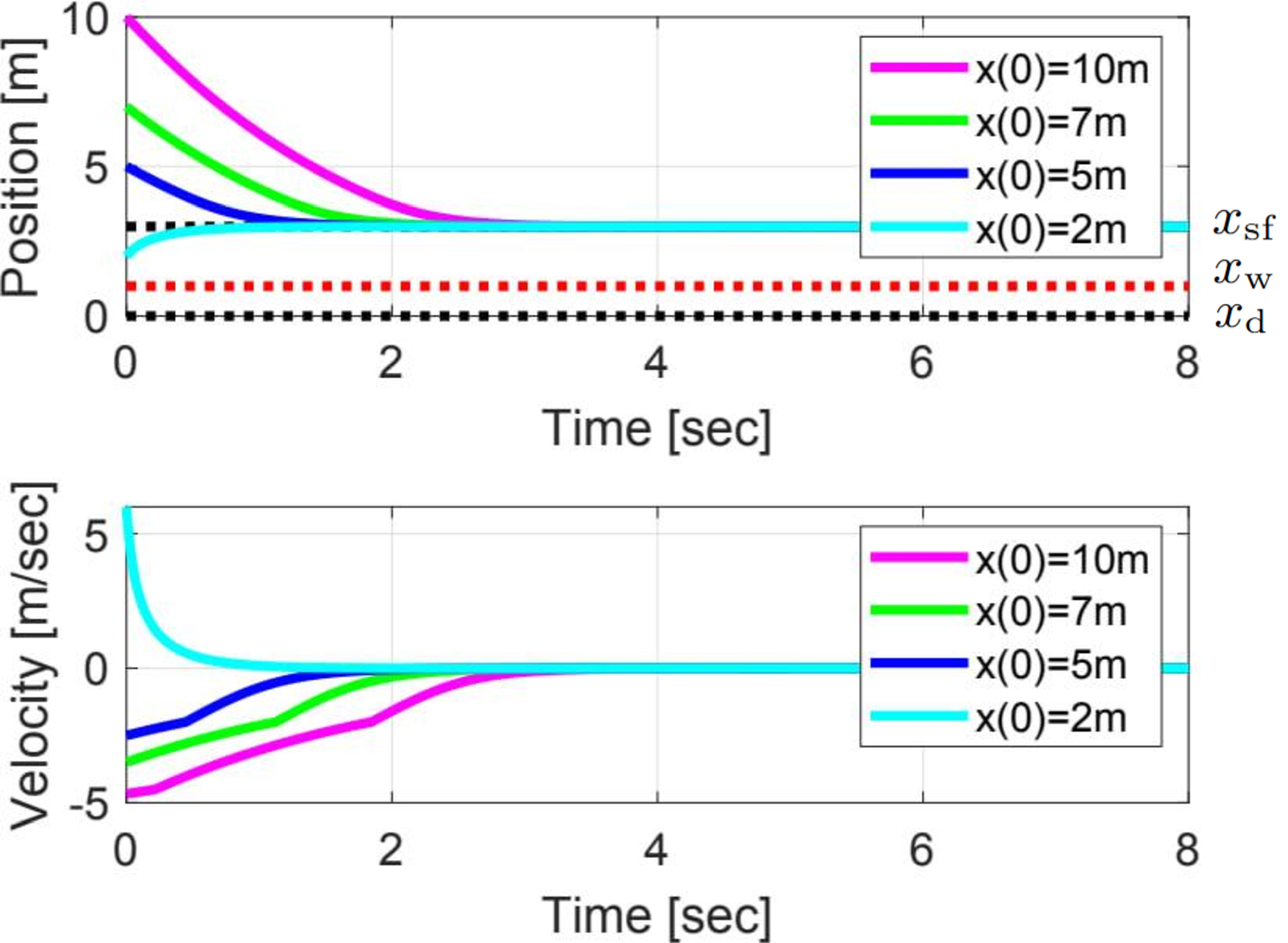}  % The printed column 
\caption{Position \& velocity for first-order graceful control} % width is 8.4 cm.
\label{fig.1stOrderGraceful1}                 % Size the figures 
\end{center}                 % accordingly.
\end{figure}
%%%%%%%%%%%%%%%%%%%%%%%%%%%%%%%%%%%%%%%%%%%%%%%%%%%%%

%%%%%%%%%%%%%%%%%%%%%%%%%%%%%%%%%%%%%%%%%%%%%%%%%%%%%
\begin{figure} 
\begin{center}
\includegraphics[width=0.45\textwidth]{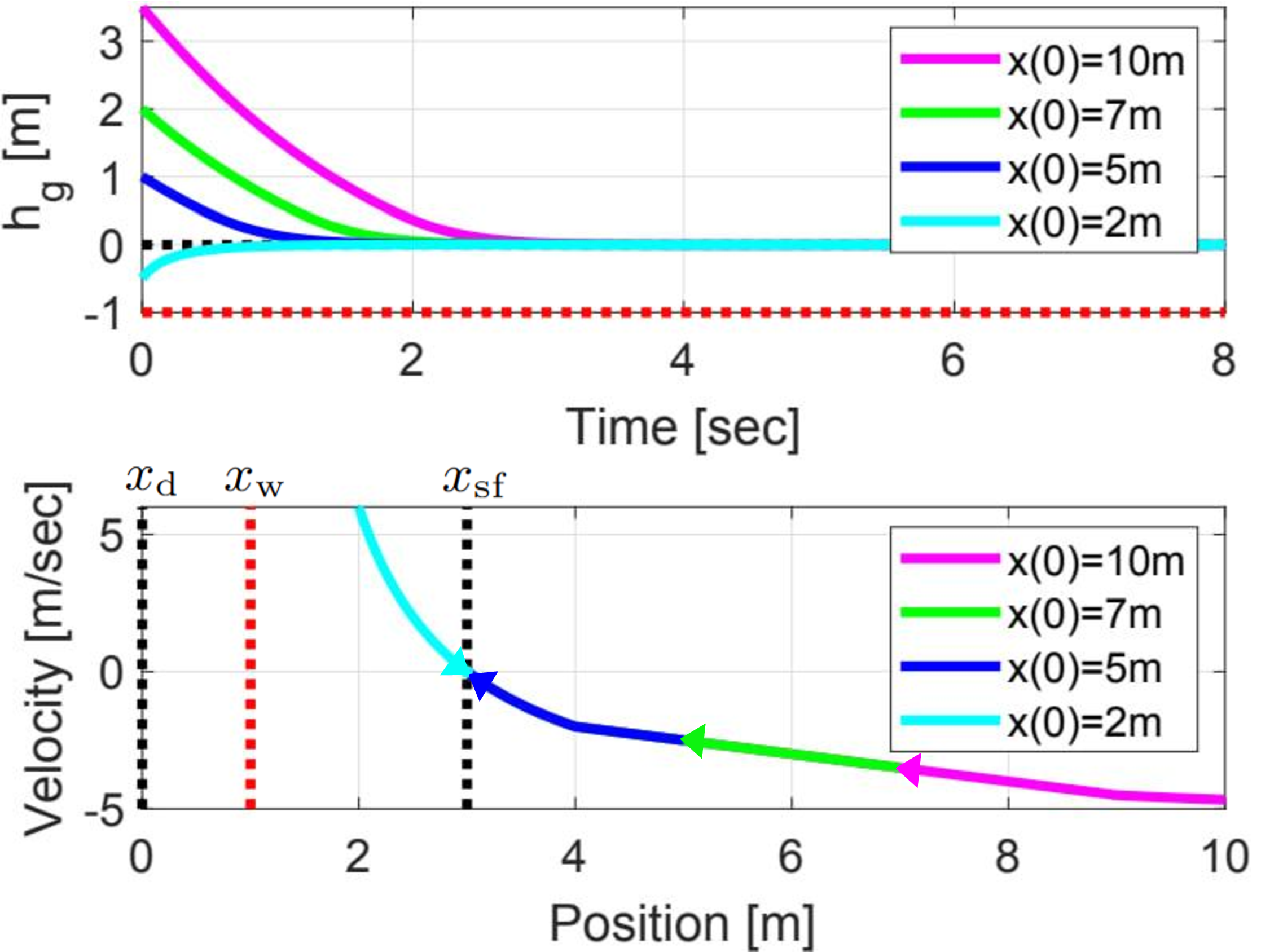}  % The printed column 
\caption{Barrier \& velocity vs. position plot for first-order graceful controller} % width is 8.4 cm.
\label{fig.1stOrderGraceful2}                 % Size the figures 
\end{center}                 % accordingly.
\end{figure}
%%%%%%%%%%%%%%%%%%%%%%%%%%%%%%%%%%%%%%%%%%%%%%%%%%%%%

Let us consider the first-order differential equation ${\dot{x}=u}$ (cf.~\eqref{eq:1storderODE}) with nominal controller
${u_\mathrm{d} = -k(x-x_\mathrm{d})}$ (cf.~\eqref{eq.1stOrderBaselineDesiredxdot}).
Then, the first-order graceful safety condition \eqref{eq.1stgraceful} with class-$\mathcal{K}$ function ${\beta(r)=\gamma r}$ yields 
\begin{equation}\label{eq.1stOrderGraceSafexdot}
    \begin{split}
        \dot{h}_g(x,u) &\geq \gamma \frac{h_g(x)}{h_g(x)+1}
        \\
        \Rightarrow \quad
        u &\geq - \gamma \frac{x-x_\mathrm{sf}}{x-x_\mathrm{w}} (x_\mathrm{sf}-x_\mathrm{w}) =: u_\mathrm{sf}.
    \end{split}
\end{equation}
Then one may set up the quadratic program \eqref{eq.1stOrderBenchmark_Optimization}
where $u_\mathrm{d}$ and $u_\mathrm{sf}$ are defined in \eqref{eq.1stOrderBaselineDesiredxdot} and \eqref{eq.1stOrderGraceSafexdot}. 
Using the KKT condition, this yields the analytical solution
\begin{equation}
\begin{split}
u^* &= \max\{u_\mathrm{d},u_\mathrm{sf}\} 
\\
&= \max\{-k(x-x_\mathrm{d}),- \gamma \frac{x-x_\mathrm{sf}}{x-x_\mathrm{w}} (x_\mathrm{sf}-x_\mathrm{w})\},
\end{split}
\end{equation}
cf.~\eqref{eq.fromKKT1}.
Substituting this into \eqref{eq:1storderODE} gives the closed-loop dynamics.

% As $h_g$ approaches -1, the denominator of the above constraint approaches 0, forcing the input to approach infinity at an extremely unsafe state. As a result, the quadratic safety control program is written as:
% \begin{equation}
% \begin{split}
%   &\min_{\dot{x}} \frac{1}{2}(\dot{x}(t)-\dot{x}_\mathrm{d}(t))^2 \\
%   &s.to:\\
%   &\dot{x}_\mathrm{d}(t)=k(x_{d}-x(t)),\\
%   &\dot{x}\geq\alpha(x_\mathrm{sf}-x_\mathrm{w})(\frac{1}{h_g+1}-1).\\
%   \label{eq.1stOrderGraceful_Optimization}
% \end{split}
% \end{equation}
% Here, the same performance objective from~\eqref{eq.1stOrderBaselineDesiredxdot} is implemented. Then, using the same values of distance parameters and eigenvalues as the benchmark case, the simulation is repeated using the above graceful safety controller. Once again, the simulation is conducted on MATLAB using the ``ode15s" function. The same 8-second time horizon with a time step of 1 ms is utilized.

Figure \ref{fig.1stOrderGraceful1} plots the position and velocity of the object, and Fig. \ref{fig.1stOrderGraceful2} depicts the graceful barrier function and position vs.\ velocity plot. 
In the $h_g$ graph, the primary safe layer lies above the black dotted line at ${h_g=0}$, while the boundary of the secondary failsafe layer is highlighted by the red dotted line at ${h_g=-1}$.
Similar to the benchmark case, the graceful controller brings the object towards $x_\mathrm{sf}$ regardless of the object's initial position. 
A noticeable difference is evident in the position vs.\ velocity plot. 
Due to the nonlinearity of the graceful barrier function, the controller implements more aggressive control input values compared to the zeroing CBF benchmark once the object approaches the catastrophic collision condition at $x_\mathrm{w}$. 
This is clearly demonstrated for ${x(0)=2\ {\rm m}}$  in Fig.~\ref{fig.1stOrderGraceful1}, where the graceful controller applies a much bigger initial velocity compared to Fig.~\ref{fig.1stOrderBenchmark1}, bringing the object back towards safety faster than the baseline case.

%%%%%%%%%%%%%%%%%%%%%%%%%%%%%%%%%%%%%%%%%%%%%%%%%%%%%%%%%%%%%%
\subsection{Scenario \#2: second-order graceful safety control example}
%%%%%%%%%%%%%%%%%%%%%%%%%%%%%%%%%%%%%%%%%%%%%%%%%%%%%%%%%%%%%%
 
We now present the second-order graceful safety control simulation study. 
Here we utilize the first-order differential equation ${\ddot{x}=u}$ (cf.~\eqref{eq:2ndorderODE}) with nominal controller
${u_\mathrm{d} = -k_1(x-x_\mathrm{d})-k_2\dot{x}}$ (cf.~\eqref{eq.DesiredAcceleration}).
%Unlike the previous section, the 2$^{nd}$ order of the system makes the object's acceleration the control input. The object's position and velocity are the state variables. 

We construct the second-order graceful CBF by utilizing the same graceful barrier function \eqref{eq.h_g} as for the first-order case. 
The second-order graceful safety controller \eqref{eq.2ndgraceful_general} yields the safety constraint
\begin{align}\label{eq.2ndOrderGraceSafexdot}
        \ddot{h}_g(x,\dot{x},u) &\geq -2\zeta\omega_\mathrm{n} \dot{h}_g(x,\dot{x}) - \omega_\mathrm{n}^2 \frac{h_g(x)}{h_g(x)+1}
        \\
        \Rightarrow \quad
        u &\geq - \omega_\mathrm{n}^2 \frac{x-x_\mathrm{sf}}{x-x_\mathrm{w}} (x_\mathrm{sf}-x_\mathrm{w}) - 2\zeta\omega_\mathrm{n} \dot{x} =: u_\mathrm{sf}. \nonumber
\end{align}
% \begin{equation}
% \begin{split}
%     &\ddot{h}_g+2\zeta\omega\dot{h}_g+\omega^2(1-\frac{1}{h_g+1})\geq0\\
%     &\ddot{x} \geq (x_\mathrm{sf}-x_\mathrm{w})(-2\zeta\omega\dot{h}_g-\omega^2(1-\frac{1}{h_g+1}))\\
% \end{split}
% \end{equation}
As a result, the optimization problem is constructed as 
\eqref{eq.1stOrderBenchmark_Optimization}
with $u_\mathrm{d}$ and $u_\mathrm{sf}$ are defined in \eqref{eq.DesiredAcceleration} and \eqref{eq.2ndOrderGraceSafexdot}. 
Then applying the KKT conditions results in
\begin{equation}
\begin{split}
u^* &= \max\{u_\mathrm{d},u_\mathrm{sf}\} 
\\
&= \max\{-k_1(x-x_\mathrm{d})-k_2\dot{x}, 
\\
&\qquad\quad\,\,\, - \omega_\mathrm{n}^2 \frac{x-x_\mathrm{sf}}{x-x_\mathrm{w}} (x_\mathrm{sf}-x_\mathrm{w}) - 2\zeta\omega_\mathrm{n} \dot{x} \},
\end{split}
\end{equation}
cf.~\eqref{eq.fromKKT2}, and substituting this into \eqref{eq:2ndorderODE} gives the closed-loop system.
% \begin{equation}
% \begin{split}
%   &\min_{\ddot{x}} \frac{1}{2}(\ddot{x}-\ddot{x}_\mathrm{d})^2 \\
%   &s.to:\\
%   &\ddot{x} \geq (x_\mathrm{sf}-x_\mathrm{w})(-2\zeta\omega\dot{h}_g-\omega^2(1-\frac{1}{h_g+1}))\\
%   &\ddot{x}_\mathrm{d} = -k_1(x-x_\mathrm{d})-k_2\dot{x}.\\
%    &h_g(x) = \frac{x-x_\mathrm{sf}}{x_\mathrm{sf}-x_\mathrm{w}}\\
%    \label{eq.2ndOrdergCBFOptimization}
% \end{split}
% \end{equation}
% Same as the benchmark controller, the optimization objective is to minimize the error between the desired and the actual acceleration of the object, subject to the safety constraint enforced by graceful safety control. In this controller, $h_g=-1$ represents the boundary between the undesirable set and the catastrophic (i.e., collision) condition. 

For the simulations the parameter $\omega_\mathrm{n}$ is set to ${2\ {\rm rad/s}}$. 
Both overdamped and underdamped controllers are simulated, using $\zeta$ values of 2 and 0.5, respectively. 
The same initial velocity value of ${-25\ {\rm m/s}}$ is implemented as for the benchmark case. 
Also, two different initial position values of ${5\ {\rm m}}$ and ${2\ {\rm m}}$ are used, which are the cases where the benchmark exponential controller failed to prevent collisions.

%%%%%%%%%%%%%%%%%%%%%%%%%%%%%%%%%%%%%%%%%%%%%%%%%%%%%
\begin{figure}
\begin{center}
\includegraphics[width=0.59\textwidth, angle=270]{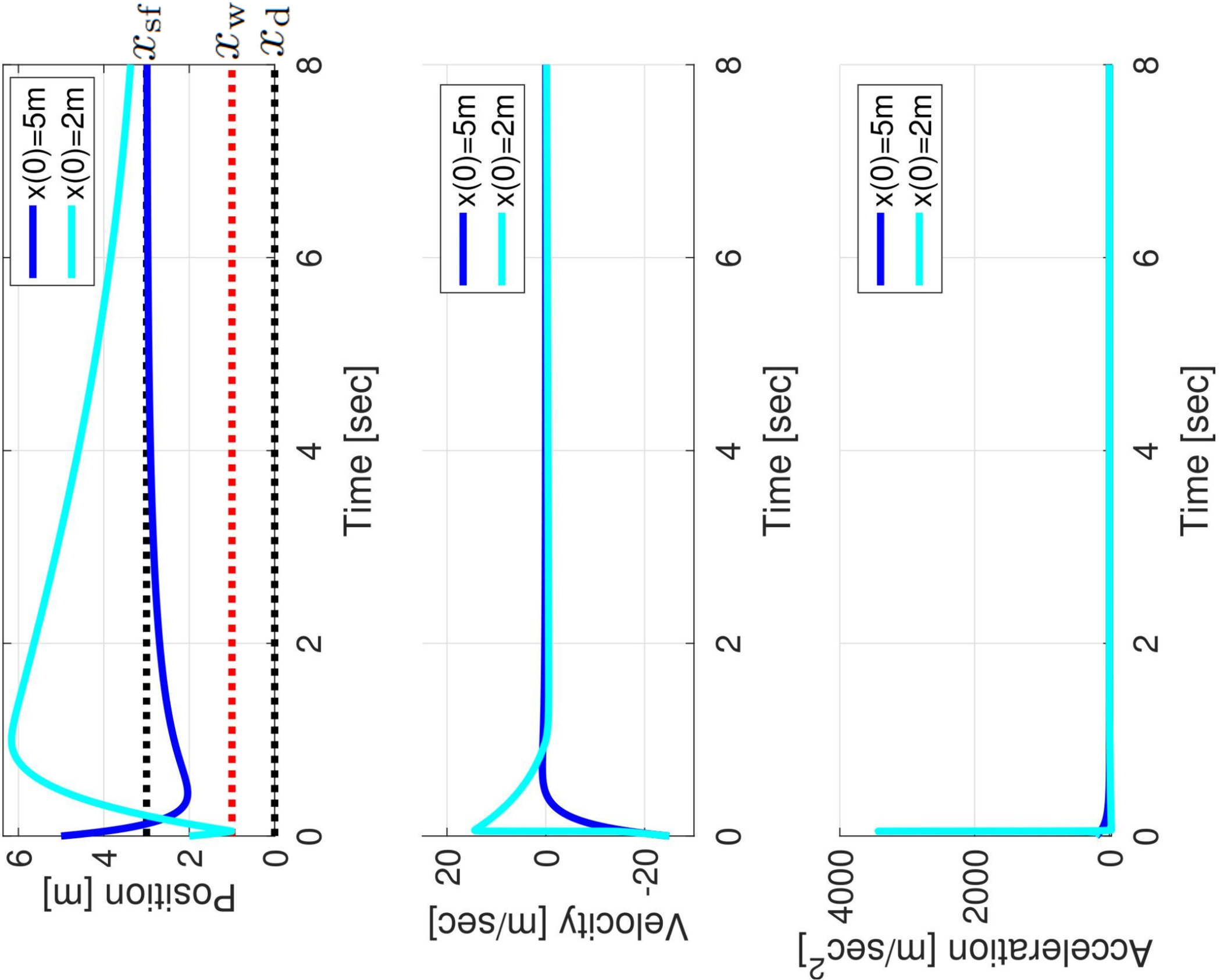}  % The printed column 
\caption{Position, velocity, \& acceleration for overdamped controller} % width is 8.4 cm.
\label{fig.overdamped1}                 % Size the figures 
\end{center}                 % accordingly.
\end{figure}
%%%%%%%%%%%%%%%%%%%%%%%%%%%%%%%%%%%%%%%%%%%%%%%%%%%%%

%%%%%%%%%%%%%%%%%%%%%%%%%%%%%%%%%%%%%%%%%%%%%%%%%%%%%
\begin{figure}
\begin{center}
\includegraphics[width=0.45\textwidth]{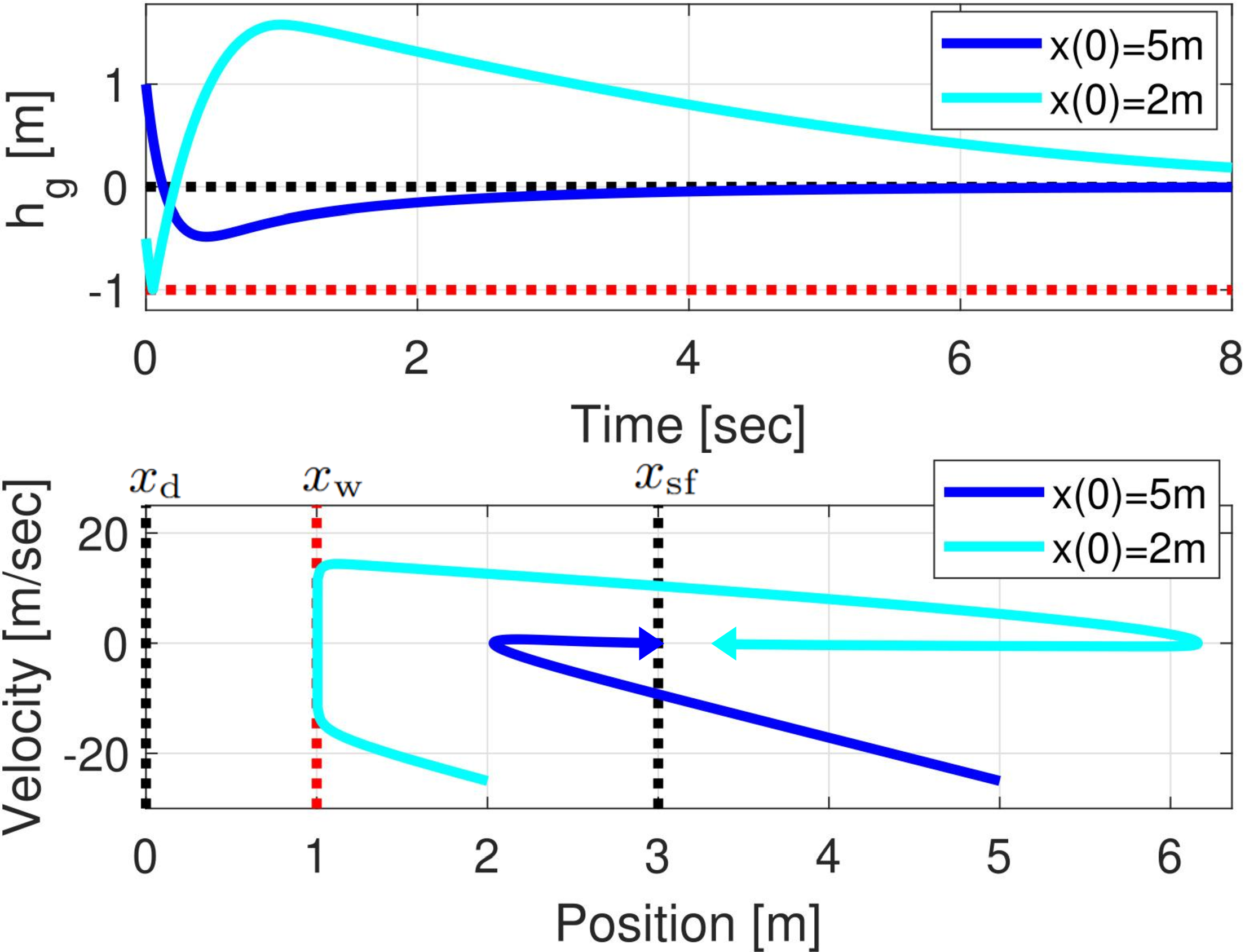}  % The printed column 
\caption{Barrier \& phase plane plot for overdamped controller} % width is 8.4 cm.
\label{fig.overdamped2}                 % Size the figures 
\end{center}                 % accordingly.
\end{figure}
%%%%%%%%%%%%%%%%%%%%%%%%%%%%%%%%%%%%%%%%%%%%%%%%%%%%%

%%%%%%%%%%%%%%%%%%%%%%%%%%%%%%%%%%%%%%%%%%%%%%%%%%%%%
\begin{figure}
\begin{center}
\includegraphics[width=0.59\textwidth, angle=270]{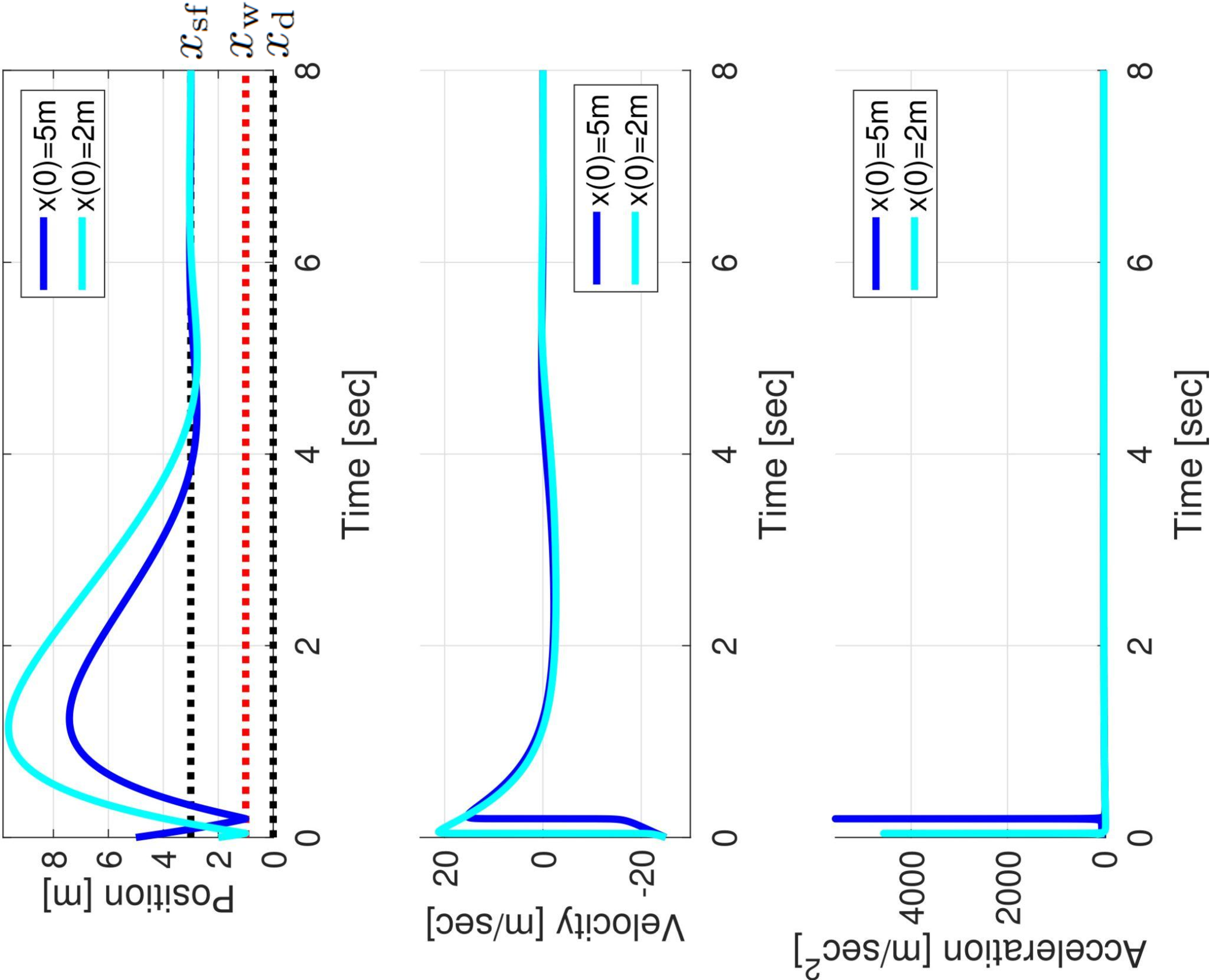}  % The printed column 
\caption{Position, velocity, \& acceleration for underdamped controller} % width is 8.4 cm.
\label{fig.underdamped1}                 % Size the figures 
\end{center}                 % accordingly.
\end{figure}
%%%%%%%%%%%%%%%%%%%%%%%%%%%%%%%%%%%%%%%%%%%%%%%%%%%%%

%%%%%%%%%%%%%%%%%%%%%%%%%%%%%%%%%%%%%%%%%%%%%%%%%%%%%
\begin{figure}
\begin{center}
\includegraphics[width=0.45\textwidth]{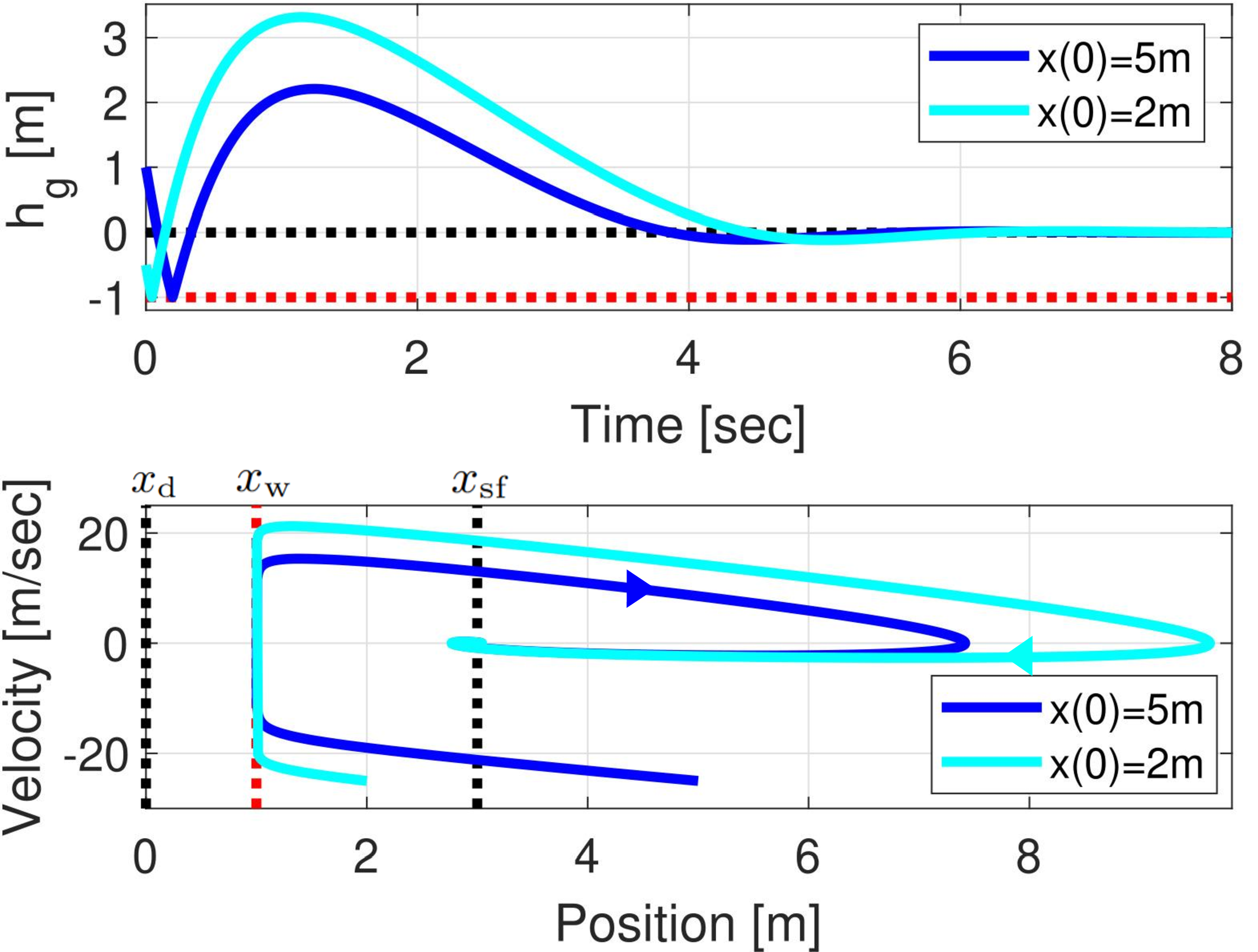}  % The printed column 
\caption{Barrier \& phase plane plot for underdamped controller} % width is 8.4 cm.
\label{fig.underdamped2}                 % Size the figures 
\end{center}                 % accordingly.
\end{figure}
%%%%%%%%%%%%%%%%%%%%%%%%%%%%%%%%%%%%%%%%%%%%%%%%%%%%%

Figures~\ref{fig.overdamped1}-\ref{fig.underdamped2} plot the simulation results for the over- and underdamped graceful safety controllers. 
In the $h_g$ plots, the black dotted line and the red dotted line represent the boundaries of the primary safe layer and the secondary failsafe layer, respectively. 

As shown in these figures, in all scenarios, the graceful safety controller successfully avoids collision and brings the object's position towards $x_\mathrm{sf}$ by utilizing very high acceleration values. 
The maximum acceleration imposed by the controller is quite large but finite (namely, ${\sim 10^{-8} \ {\rm m/s}^2}$). 
This is consistent with the fact that the catastrophic boundary at ${h_g = -1}$, which corresponds to an infinite value of the Lyapunov candidate function used earlier in this article for analyzing the properties of the controller, is approached but never reached.

At least three interesting distinctions are visible between the under- and overdamped graceful safety controllers. 
First, the overdamped controller always returns to safety exponentially, whereas the underdamped controller exhibits more oscillatory behavior, as expected. 
Second, the overdamped controller applies peak accelerations of ${3400\ {\rm m/s}^2}$ for ${x(0)=2\ {\rm m}}$ and ${200\ {\rm m/s}^2}$ for ${x(0)=5\ {\rm m}}$. 
In contrast, the underdamped controller applies peak accelerations of ${4500\ {\rm m/s}^2}$ for ${x(0) = 2\ {\rm m}}$ and ${5500\ {\rm m/s}^2}$ for ${x(0) = 5\ {\rm m}}$. 
Third, the underdamped controller returns the system to safety faster than the overdamped controller. 
These tradeoffs, while may not be surprising, highlight the potential value of treating the graceful safety controller's damping ratio as one of multiple tunable controller design parameters. 
Overall, these simulation results successfully prove the effectiveness of the proposed controller and motivate the need for ``grace" in safety-critical control.

%%%%%%%%%%%%%%%%%%%%%%%%%%%%%%%%%%%%%%%%%%%%%%%%%%%%%%%%%%%%%%
\section{Conclusion}\label{sec:conclusion}
%%%%%%%%%%%%%%%%%%%%%%%%%%%%%%%%%%%%%%%%%%%%%%%%%%%%%%%%%%%%%%

Control barrier functions (CBFs) are widely used in the safe control literature due to their rigorous safety guarantee and natural switching control ability. 
However, the existing literature on CBFs creates a single definition of what it means to be in an unsafe state, while there exist multiple layers of safety in practice. 
To address this gap, a graceful safety control framework was introduced in this paper. 
The proposed framework utilizes a nonlinear barrier function to define a primary desired safety layer and a secondary failsafe layer, then imposes zeroing and reciprocal CBF-like conditions on them, respectively. 
The effectiveness of the proposed controller is demonstrated via mathematical proofs and simulation studies. 
In the present study, only first and second-order graceful safety control was considered for simplicity. 
One potential topic for future work is generalizing these findings to systems with even higher relative degrees, thereby extending its applicability and versatility.

%%%%%%%%%%%%%%%%%%%%%%%%%%%%%%%%%%%%%%%%%%%%%%%%%%%%%%%%%%%%%%
\section*{Acknowledgment}
%%%%%%%%%%%%%%%%%%%%%%%%%%%%%%%%%%%%%%%%%%%%%%%%%%%%%%%%%%%%%%

Support for this research was provided by the U.S. National Science Foundation (NSF). 
The authors gratefully acknowledge this support. 
This paper reflects the opinions of the authors, not NSF.

\bibliographystyle{plain}    % Include this if you use bibtex 
\bibliography{autosam}      % and a bib file to produce the 
                 % bibliography (preferred). The
                 % correct style is generated by
                 % Elsevier at the time of printing.

%\begin{thebibliography}{99}   % Otherwise use the 
                 % thebibliography environment.
                 % Insert the full references here.
                 % See a recent issue of Automatica 
                 % for the style.
% \bibitem[Heritage, 1992]{Heritage:92}
%   (1992) {\it The American Heritage. 
%   Dictionary of the American Language.}
%   Houghton Mifflin Company.
% \bibitem[Able, 1956]{Abl:56}
%   B.~C.~Able (1956). Nucleic acid content of macroscope. 
%   {\it Nature 2}, 7--9. 
% \bibitem[Able {\em et al.}, 1954]{AbTaRu:54}  
%   B.~C. Able, R.~A. Tagg, and M.~Rush (1954).
%   Enzyme-catalyzed cellular transanimations.
%   In A.~F.~Round, editor, 
%   {\it Advances in Enzymology Vol. 2} (125--247). 
%   New York, Academic Press.
% \bibitem[R.~Keohane, 1958]{Keo:58}
%   R.~Keohane (1958).
%   {\it Power and Interdependence: 
%   World Politics in Transition.}
%   Boston, Little, Brown \& Co.
% \bibitem[Powers, 1985]{Pow:85}
%   T.~Powers (1985).
%   Is there a way out?
%   {\it Harpers, June 1985}, 35--47.

%\end{thebibliography}

\end{document}